\documentclass[times, 11pt]{article} 
\usepackage{times}
\usepackage{amssymb}
\usepackage{amsthm}
\usepackage{amsmath}
\usepackage{graphicx}
\usepackage{latex8}
\usepackage{mathptmx}
\usepackage[scaled]{helvet}
\usepackage{txfonts}
\usepackage{algorithmic}
\usepackage[T1]{fontenc}
\usepackage{algorithm}
\usepackage{algorithmic}

\newtheorem{conjecture}{Conjecture}
\newtheorem{case}{Case}
\newenvironment{proofsketch}{\noindent {\em {Proof (sketch):}}}{$\Box$\vskip \belowdisplayskip}

\newenvironment{theorem1}[2]{\noindent \textbf{{{#1}~\ref{#2}.}}}{}

\def\cR{\mathcal{R}}
\def\cP{\mathcal{P}}

\newcommand{\set}[1]{\left\{#1\right\}}

\begin{document}


\hbadness=10000
\vbadness=10000

\title{\LARGE Vertex Sparsifiers and Abstract Rounding Algorithms}

\author{Moses Charikar \thanks{moses@cs.princeton.edu, Center for Computational Intractability,
Department of Computer Science, Princeton University,
supported by NSF awards MSPA-MCS 0528414, CCF 0832797, and AF 0916218} \and Tom Leighton \thanks{ftl@math.mit.edu, Mathematics Department, MIT and Akamai Technologies, Inc.} \and Shi Li \thanks{shili@cs.princeton.edu, Center for Computational Intractability,
Department of Computer Science, Princeton University,
supported by NSF awards MSPA-MCS 0528414, CCF 0832797, and AF 0916218.} \and Ankur Moitra \thanks{moitra@mit.edu, Computer Science and Artificial Intelligence Laboratory,  MIT, This research was supported in part by a Fannie and John Hertz Foundation Fellowship. Part of this work was done while the author was visiting Princeton University.}  }

\maketitle
\thispagestyle{empty}

\begin{abstract}
\emph{The notion of vertex sparsification (in particular cut-sparsification) is introduced in \cite{M}, where it was shown that for any graph $G = (V, E)$ and a subset of $k$ terminals $K \subset V$, there is a polynomial time algorithm to construct a graph $H = (K, E_H)$ \emph{on just the terminal set} so that simultaneously for all cuts $(A, K-A)$, the value of the minimum cut in $G$ separating $A$ from $K -A$ is approximately the same as the value of the corresponding cut in $H$. Then approximation algorithms can be run directly on $H$ as a proxy for running on $G$, yielding approximation guarantees independent of the size of the graph. In this work, we consider how well cuts in the sparsifier $H$ can approximate the minimum cuts in $G$, and whether algorithms that use such reductions need to incur a multiplicative penalty in the approximation guarantee depending on the quality of the sparsifier.}

\emph{We give the first super-constant lower bounds for how well a cut-sparsifier $H$ can simultaneously approximate all minimum cuts in $G$. 
We prove a lower bound of
$\Omega(\log^{1/4} k)$ -- this is polynomially-related to the known upper
bound of $O(\log k/\log \log k)$.
This is an exponential improvement on the $\Omega(\log \log k)$ bound given in \cite{LM} which in fact was for a stronger vertex sparsification guarantee, and did not apply to cut sparsifiers.}

\emph{Despite this negative result, we show that for many natural problems, we do not need to incur a multiplicative penalty for our reduction. Roughly, we show that any rounding algorithm which also works for the $0$-extension relaxation can be used to construct good vertex-sparsifiers for which the optimization problem is easy. Using this, we obtain optimal $O(\log k)$-competitive Steiner oblivious routing schemes, which generalize the results in \cite{R}. We also demonstrate that for a wide range of graph packing problems (which includes maximum concurrent flow, maximum multiflow and multicast routing, among others, as a special case), the integrality gap of the linear program is always at most $O(\log k)$ times the integrality gap restricted to trees. This result helps to explain the 
ubiquity of the $O(\log k)$ guarantees for such problems.
Lastly, we use our ideas to give an efficient construction for vertex-sparsifiers that match the current best existential results -- this was previously open. Our algorithm makes novel use of Earth-mover constraints.}

\end{abstract}

\newpage

\section{Introduction} 


\subsection{Background}

The notion of vertex sparsification (in particular cut-sparsification) is
introduced in \cite{M}: Given a graph $G = (V, E)$ and a subset of terminals $K
\subset V$, the goal is to construct a graph $H = (K, E_H)$ \emph{on just the
terminal set} so that simultaneously for all cuts $(A, K-A)$, the value of the
minimum cut in $G$ separating $A$ from $K -A$ is approximately the same as the
value of the corresponding cut in $H$. If for all cuts $(A, K - A)$, the the
value of the cut in $H$ is at least the value of the corresponding minimum cut
in $G$ and is at most $\alpha$ times this value, then we call $H$ a
cut-sparsifier of quality $\alpha$. 

The motivation for considering such questions is in obtaining approximation
algorithms with guarantees that are independent of the size of the graph. For
many graph partitioning and multicommodity flow questions, the value of the
optimum solution can be approximated given just the values of the minimum cut
separating $A$ from $K -A$ in $G$ (for \emph{every} $A \subset K$). As a result
the value of the optimum solution is approximately preserved, when mapping the
optimization problem to $H$. So approximation algorithms can be run on $H$ as a
proxy for running directly on $G$, and because the size (number of nodes) of $H$
is $|K|$, any approximation algorithm that achieves a $poly(\log
|V|)$-approximation guarantee in general will achieve a $poly(\log |K|)$
approximation guarantee when run on $H$ (provided that the quality $\alpha$ is
also $poly(\log |K|)$). Feasible solutions in $H$ can also be mapped back to
feasible solutions in $G$ for many of these problems, so polynomial 
 time constructions for good cut-sparsifiers yield black box techniques for designing approximation algorithms with guarantees $poly(\log |K|)$ (and independent of the size of the graph). 

In addition to being useful for designing approximation algorithms with improved
guarantees, the notion of cut-sparsification is also a natural generalization of
many methods in combinatorial optimization that attempt to preserve certain cuts
in $G$ (as opposed to all minimum cuts) in a smaller graph $H$ - for example
Gomory-Hu Trees, and Mader's Theorem. Here we consider a number of questions
related to cut-sparsification: 

\begin{enumerate}
\itemsep=0pt
\item Is there a super-constant lower bound on the quality of cut-sparsifiers? Do the best (or even near-best) cut-sparsifiers necessarily result from (a distribution on) contractions?

\item Do we really need to pay a price (in the approximation guarantee) when applying vertex sparsification to an optimization problem?

\item Can we construct (in polynomial time) cut-sparsifiers with quality as good as the current best \emph{existential} results?

\end{enumerate}

We resolve all of these questions in this paper. In the preceding subsections, we will describe what is currently known about each of these questions, our results, and our techniques. 
\footnote{
Recently, it has come to our attention that, 
independent of and concurrent to our work,
Makarychev and Makarychev, and independently,
Englert, Gupta, Krauthgamer, Raecke, Talgam and Talwar
obtained results similar to some in this paper.
}

\subsection{Super-Constant Lower Bounds and Separations}

In \cite{M}, it is proven that in general there are always cut-sparsifiers $H$ of quality at most $O(\log k / \log \log k)$. In fact, if $G$ excludes any fixed minor then this bound improves to $O(1)$. Yet prior to this work, no super-constant lower bound was known for the quality of cut-sparsifiers in general. We prove

\begin{theorem}~\label{thm:lower}
There is an infinite family of graphs that admits no cut-sparsifiers of quality better than $\Omega(\log^{1/4} k)$. 
\end{theorem}

Some results are known in more general settings. In particular, one could require that the graph $H$ not only approximately preserve minimum cuts but also approximately preserve the congestion of all multicommodity flows (with demands endpoints restricted to be in the terminal set). This notion of vertex-sparsification is referred to as flow-sparsification (see \cite{LM}) and admits a similar definition of quality. 
\cite{LM} gives a lower bound of $\Omega(\log \log k)$ for the quality of flow-sparsifiers.
However, this does not apply to cut sparsifiers and in fact, for the example given in \cite{LM}, there is an $O(1)$-quality cut-sparsifier!

Additionally, there are examples in which cuts can be preserved within a
constant factor, yet flows cannot: Benczur and Karger \cite{BK} proved that
given any graph on $n$ nodes, there is a sparse (weighted) graph $G'$ that approximate all cuts in $G$ within a multiplicative $(1 + \epsilon)$
factor, but one
provably cannot preserve the congestion of all multicommodity flows within a factor better than
$\Omega(\frac{\log n}{\log \log n})$ on a sparse graph (consider the complete graph $K_n$). So here the limits of sparsification are much different
for cuts than for flows.

In this paper, we give a super-constant lower bound on the quality of cut-sparsifiers in general and in fact this implies a stronger lower bound than is given in \cite{LM}. Our bound is polynomially related to the current best upper-bound, which is $O(\log k / \log \log k)$. 

We note that the current best upper bound is actually a reduction from the upper bound on the integrality gap of a particular LP relaxation for the $0$-extension problem \cite{CKR}, \cite{FHRT}. The integrality gap of this LP relaxation is known to be $\Omega(\sqrt{\log k})$. Yet, the best lower bound we are able to obtain here is $\Omega(\log^{1/4} k)$. This leads us to our next question: Do integrality gaps for the $0$-extension LP immediately imply lower bounds for cut-sparsification? This question, as we will see, is essentially equivalent to the question of whether or not the best cut-sparsifiers necessarily come from a distribution on contractions. 

Lower bounds on the quality of cut-sparsifiers (in this paper) and flow-sparsifiers (\cite{LM}) are substantially more complicated than integrality gap examples for the $0$-extension LP relaxation. If the best cut-sparsifiers or flow-sparsifiers were actually always generated from some distribution on contractions in the original graph via strong duality (see Section $3$), any integrality gap would immediately imply a lower bound for cut-sparsificatin or flow-sparsification. But as we demonstrate here, this is not the case:

\begin{theorem}~\label{thm:sep}
There is an infinite family of graphs so that the ratio of the best quality
cut-sparsifier to the best quality cut-sparsifier that can be achieved through a
distribution on contractions is $o(1) = O(\frac{\log\log\log\log k}{\log^2 \log
\log k})$
\end{theorem}

We also note that in order to prove this result we establish a somewhat surprising connection between cut-sparsification and the harmonic analysis of Boolean functions. The particular cut-sparsifier that we construct in order to prove this result is inspired by the noise stability operator, and as a result, we can use 
tools from harmonic analysis (Bourgain's Junta Theorem \cite{Bou} and the Hypercontractive Inequality \cite{Bon}, \cite{Bec}) to analyze the quality of the cut-sparsifier. Casting this question of bounding the quality as a question in harmonic analysis allows us to reason about many cuts simultaneously without worrying about the messy details of the combinatorics. 


\subsection{Abstract Integrality Gaps and Rounding Algorithms}

As described earlier, running an approximation algorithm on the sparsifier $H = (K, E_H)$ as a proxy for the graph $G = (V, E)$ pays an additional price in the approximation guarantee that corresponds to how well $H$ approximates $G$. Here we consider the question of whether this loss can be avoided. 

As a motivating example, consider the problem of Steiner oblivious routing \cite{M}. Previous techniques for constructing Steiner oblivious routing schemes \cite{M}, \cite{LM} first construct a flow-sparsifier $H$ for $G$, construct an oblivious routing scheme in $H$ and then map this back to a Steiner oblivious routing scheme in $G$. Any such approach must pay a price in the competitive ratio, and cannot achieve  an $O(\log k)$-competitive guarantee because (for example) expanders do not admit constant factor flow-sparsifiers \cite{LM}.

So black box reductions pay a price in the competitive ratio, yet here we present a technique for \emph{combining} the flow-sparsification techniques in \cite{LM} and the oblivious routing constructions in \cite{R} into a single step, and we prove that there are $O(\log k)$-competitive Steiner oblivious routing schemes, which is optimal. This result is a corollary of a more general idea:

The constructions of flow-sparsifiers given in \cite{LM} (which is an extension of the techniques in \cite{M}) can be regarded as a dual to the rounding algorithm in \cite{FHRT} for the $0$-extension problem. What we observe here is: Suppose we are given a rounding algorithm that is used to round the fractional solution of some relaxation to an integral solution for some optimization problem. If this rounding algorithm also works for the relaxation for the $0$-extension problem given in \cite{K} (and also used in \cite{CKR}, \cite{FHRT}), then we can use the techniques in \cite{M}, \cite{LM} to obtain \emph{stronger} flow-sparsifiers which are not only good quality flow-sparsifiers, but also for which the optimization problem is easy. So in this way we do not need to pay an additional price in the approximation guarantee in order to replace the dependence on $n$ with a dependence on $k$. With these ideas in mind, what we observe is that the rounding algorithm in \cite{FRT} wh
 ich embed
 s metric spaces into distributions on dominating tree-metrics, can also be used to round the $0$-extension relaxation. This allows us to construct flow-sparsifiers that have $O(\log k)$-quality, and also can be explicitly written as a convex combination of $0$-extensions that are tree-like. On trees, oblivious routing is easy, and so this gives us a way to simultaneously construct good flow-sparsifiers and good oblivious routing schemes on the sparsifier in one step! 

Of course, the rounding algorithm in \cite{FRT} for embedding metric spaces into distributions on dominating tree-metrics is a \emph{very} common first step in rounding fractional relaxations of graph partitioning, graph layout and clustering problems. So for all problems that use this embedding as the main step, we are able to replace the dependence on $n$ with dependence on $k$, and we do not introduce any additional poly-logarithmic factors as in previous work! One can also interpret our result as giving a generalization of the hierarchical decompositions given in \cite{R} for approximating the cuts in a graph $G$ on trees. We state our results more formally, below, and we refer to such a statement as an {\sc Abstract Integrality Gap}.  

\begin{definition}~\label{def:gpp}
We call a fractional packing problem $P$ a graph packing problem if the goal of the dual covering problem $D$ is to minimize the ratio of the total units of distance $\times$ capacity allocated in the graph divided by some monotone increasing function of the distances between terminals.
\end{definition}

This definition is quite general, and captures maximum concurrent flow, maximum multiflow, and multicast routing as special cases, in addition to many other common optimization problems. The integral\footnotemark[1] dual $ID$ problems are generalized sparsest cut, multicut and requirement cut respectively.

\footnotetext[1]{The notion of what constitutes an integral solution depends on the problem. In some cases, it translates to the distances are all $0$ or $1$, and in other cases it can mean something else. The important point is that the notion of integral just defines a class of admissible metrics, as opposed to arbitrary metrics which can arise in the packing problem.}

\begin{theorem}~\label{thm:gpp}
For any graph packing problem $P$, the maximum ratio of the integral dual to the fractional primal is at most $O(\log k)$ times the maximum ratio restricted to trees. 
\end{theorem}

%
For a packing problem that fits into this class,
this theorem allows us to reduce bounding the integrality gap in general graphs to bounding the integrality gap on trees, which is often substantially easier than for general graphs (i.e. for the example problems given above). 
We believe that this result helps to explain the intrinsic robustness of fractional packing problems into undirected graphs,
in particular the ubiquity of the $O(\log k)$ bound for
the flow-cut gap for a wide range of multicommodity flow problems.

We also give a polynomial time algorithm to reduce any graph packing problem $P$ to a corresponding problem on a tree: Again, let $K$ be the set of terminals.

\begin{definition}
Let $OPT(P, G)$ be the optimal value of the fractional graph packing problem $P$ on the graph $G$.
\end{definition}

\begin{theorem}~\label{thm:agpp}
There is a polynomial time algorithm to construct a distribution $\mu$ on (a polynomial number of) trees on the terminal set $K$, s.t. $$E_{T \leftarrow \mu}[OPT(P, T)] \leq O(\log k) OPT(P, G)$$ and such that any valid integral dual of cost $C$ (for any tree $T$ in the support of $\mu$) can be immediately transformed into a valid integral dual in $G$ of cost at most $C$. 
\end{theorem}

As a corollary, given an approximation algorithm that achieves an approximation ratio of $C$ for the integral dual to a graph packing problem on trees, we obtain an approximation algorithm with a guarantee of $O(C \log k)$ for general graphs. We will refer to this last result as an {\sc Abstract Rounding Algorithm}.

We also give a polynomial time construction of $O(\log k/ \log \log k)$ quality flow-sparsifiers (and consequently cut-sparsifiers as well), which were previously only known to exist, but finding a polynomial time construction was still open. We accomplish this by performing a lifting (inspired by Earth-mover constraints) on an appropriate linear program. This lifting allows us to implicitly enforce a constraint that previously was difficult to enforce, and required an approximate separation oracle rather than an exact separation oracle.  We give the details in section ~\ref{sec:alift}. 


\section{Maximum Concurrent Flow}

An instance of the maximum concurrent flow problem consists of an undirected graph $G = (V, E)$, a capacity function $c: E \rightarrow \Re^+ $ that assigns a non-negative capacity to each edge, and a set of demands $\{ (s_i, t_i, f_i)\}$ where $s_i, t_i \in V$ and $f_i$ is a non-negative demand. We denote $K = \cup_i \{s_i, t_i\}$. The maximum concurrent flow question asks, given such an instance, what is the largest fraction of the demand that can be simultaneously satisfied? This problem can be formulated as a polynomial-sized linear program, and hence can be solved in polynomial time. However, a more natural formulation of the maximum concurrent flow problem can be written using an exponential number of variables. 

For any $a, b \in V$ let $P_{a, b}$ be the set of all (simple) paths from $a$ to $b$ in $G$. Then the maximum concurrent flow problem and the corresponding dual can be written as :

\[ \begin{array}{rlrl}
    \max        & \lambda  \hspace{5.0pc} & \hspace{4.0pc} \min & \sum_{e} d(e) c(e)\\
    \mbox{s.t.} & & \hspace{4.0pc}  \mbox{s.t.} & \\
                &  \sum_{P \in P_{s_i, t_i}} x(P) \geq \lambda f_i & \hspace{4.0pc} &  \forall_{P \in P_{s_i, t_i}} \sum_{e \in P} d(e) \geq D(s_i, t_i)  \\
                &  \sum_{P \owns e} x(P) \leq c(e) & \hspace{4.0pc} & \sum_i D(s_i, t_i) f_i \geq 1\\
                &  x(P) \geq 0 & \hspace{4.0pc} & d(e) \geq 0, D(s_i, t_i) \geq 0
    \end{array}
\]

For a maximum concurrent flow problem, let $\lambda^*$ denote the optimum.

Let $|K| = k$. Then for a given set of demands $\{s_i, t_i, f_i\}$, we associate
a vector $\vec{f} \in \Re^{k \choose 2}$ in which each coordinate corresponds to
a pair $(x, y) \in {K \choose 2}$ and the value $\vec{f}_{x, y}$ is defined as
the demand $f_i$ for the terminal pair $s_i = x, t_i = y$. 

\begin{definition}
We denote $cong_G(\vec{f}) = \frac{1}{\lambda^*}$
\end{definition}
Or equivalently $cong_G(\vec{f})$ is the minimum $C$ s.t. $\vec{f}$ can be routed in $G$ and the total flow on any edge is at most $C$ times the capacity of the edge.

Throughout we will use the notation that graphs $G_1, G_2$ (on the same node set) are "summed" by taking the union of their edge set (and allowing parallel edges). 


\subsection{Cut Sparsifiers}

Suppose we are given an undirected, capacitated graph $G = (V, E)$ and a set $K \subset V$ of terminals of size $k$. Let $h: 2^V \rightarrow \Re^+$ denote the cut function of $G$: $h(A) = \sum_{(u, v) \in E \mbox{ s.t. } u \in A, v \in V -A} c( u, v)$. We define the function $h_K : 2^K \rightarrow \Re^+$ which we refer to as the terminal cut function on $K$:
$h_K(U) = \min_{A \subset V \mbox{ s.t. } A \cap K = U} h(A)$.

\begin{definition}
$G'$ is a {\em cut-sparsifier} for the graph $G = (V, E)$ and the terminal set $K$ if $G'$ is a graph on just the terminal set $K$ (i.e. $G' = (K, E')$) and if the cut function $h' : 2^K \rightarrow \Re^+$ of $G'$ satisfies (for all $U \subset K$)
$$h_K(U) \leq h'(U)$$ 
\end{definition}

We can define a notion of quality for any particular cut-sparsifier:

\begin{definition}
The {\em quality} of a cut-sparsifier $G'$ is defined as $$max_{U \subset K} \frac{h'(U) }{h_K(U)}$$
\end{definition}

We will abuse notation and define $\frac{0}{0} = 1$ so that when $U$ is disconnected from $K - U$ in $G$ or if $U = \emptyset$ or $U = K$, the ratio of the two cut functions is $1$ and we ignore these cases when computing the worst-case ratio and consequently the quality of a cut-sparsifier.


\subsection{$0$-Extensions}

\begin{definition}
$f: V \rightarrow K$ is a $0$-extension if for all $a \in K$, $f(a) = a$.
\end{definition}

So a $0$-extension $f$ is a clustering of the nodes in $V$ into sets, with the property that each set contains exactly one terminal. 

\begin{definition}
Given a graph $G = (V, E)$ and a set $K \subset V$, and $0$-extension $f$, $G_f = (K, E_f)$ is a capacitated graph in which for all $a, b \in K$, the capacity $c_f(a,b)$ of edge $(a, b) \in E_f$ is $$\sum_{(u,v) \in E \mbox{ s.t. } f(u) = a, f(v) = b} c(u, v)$$
\end{definition}


\section{Lower Bounds for Cut Sparsifiers}\label{sec:lb}

Consider the following construction for a graph $G$. Let $Y$ be the hypercube of size $2^d$ for $d = \log k$. Then for every node $y_s \in Y$ (i.e. $s \in \{0, 1\}^d$), we add a terminal $z_s$ and connect the terminal $z_s$ to $y_s$ using an edge of capacity $\sqrt{d}$. All the edges in the hypercube are given capacity $1$.  We'll use this instance to show 2 lower bounds, one for 0-extension cut sparsifiers and the other for arbitrary cut sparisifers.

\subsection{Lower bound for Cut Sparsifiers from 0-extensions}

In this subseciton, we give an $\Omega(\sqrt{d})$ integrality gap for the semi-metric relaxation of the $0$-extension problem on this graph, even when the semi-metric (actually on all of $V$) is $\ell_1$. Such a bound is actually implicit in the work of \cite{JLS} too. Also , we show a strong duality between the worst case integrality gap for the semi-metric relaxation (when the semi-metric on $V$ must be $\ell_1$) and the quality of the best cut-sparsifer that can result from contractions. This gives an $\Omega(\sqrt{\log k})$ lower bound on how well a distribution on $0$-extensions can approximate the minimum cuts in $G$.

Also, given the graph $G = (V, E)$ a set $K \subset V$ of terminals, and a semi-metric $D$ on $K$ we define the $0$-extension problem as:

\begin{definition}
The \textbf{0-Extension Problem} is defined as $$\min_{\mbox{0-Extensions} f} \sum_{(u,v) \in E} c(u,v) D(f(a), f(b))$$
We denote $OPT(G, K, D)$ as the value of this optimum.
\end{definition}

\begin{definition}
Let $\Delta_U$ denote the cut-metric in which $\Delta_U (u,v) = 1_{|U \cap \{u, v\} |=1}$.
\end{definition}

Also, given an partition $\cP$ of $V$, we will refer to $\Delta_{\cP}$ as the partition metric (induced by $\cP$) which is $1$ if $u$ and $v$ are contained in different subsets of the partition $\cP$, and is $0$ otherwise. 

\[ \begin{array}{rl}
    \min & \sum_{(u, v) \in E} c(u, v) \delta(u, v) \\
    \mbox{s.t.} & \\
                &   \delta \mbox{  is a semi-metric on } V\\
                &  \forall_{t, t' \in K} \delta(t, t') = D(t, t').
    \end{array}
\]

We refer to this linear program as the \textbf{Semi-Metric Relaxation}. For a particular instance $(G, K, D)$ of the $0$-extension problem, we denote the optimal solution to this linear program as $OPT_{sm}(G, K, D)$. 

 \vspace{0.5pc} \begin{theorem}~\label{thm:fhrt}
\cite{FHRT} $$OPT_{sm}(G, K, D) \leq OPT \leq O(\frac{\log k}{\log \log k}) OPT_{sm}(G, K, D)$$
\end{theorem} \vspace{0.5pc} 

If we are given a semi-metric $D$ which is $\ell_1$, we can additionally define a stronger (exponentially) sized linear program.

\[ \begin{array}{rl}
    \min & \sum_{U} \delta(U) h(U) \\
    \mbox{s.t.} & \\
                &  \forall_{t, t' \in K} \sum_{U}\delta(U) \Delta_U(t, t') = D(t, t').
    \end{array}
\]

We will refer to this linear program as the \textbf{Cut-Cut Relaxation}. For a particular instance $(G, K, D)$ of the $0$-extension problem, we denote the optimal solution to this linear program as $OPT_{cc}(G, K, D)$. 

The value of this linear program is that an upper bound on the integrality gap of this linear program (for a particular graph $G$ and a set of terminals $K$) gives an upper bound on the quality of cut-sparsifiers. In fact, a stronger statement is true, and the quality of the best cut-sparsifier that can be achieved through contractions will be exactly equal to the maximum integrality gap of this linear program. The upper bound is given in \cite{M} -and here we exhibit a strong duality:

\begin{definition}
The \em{Contraction Quality} of $G, K$ is defined to be the minimum $\alpha$ such that there is a distribution on $0$-extensions $\gamma$ and $H = \sum_f \gamma(f) G_f$ is a $\alpha$ quality cut-sparsifier.
\end{definition}

\begin{lemma}
Let $\nu$ be the maximum integrality gap of the Cut-Cut Relaxation for a particular graph $G = (V, E)$, a particular set $K \subset V$ of terminals, over all $\ell_1$ semi-metrics $D$ on $K$. Then the Contraction Quality of $G, K$ is exactly $\nu$.
\end{lemma}

\begin{proof}
Let $\alpha$ be the Contraction Quality of $G, K$. Then implicitly in \cite{M}, $\alpha \leq \nu$. Suppose $\gamma$ is a distribution on $0$-extensions s.t. $H = \sum_f \gamma(f) G_f$ is a $\alpha$-quality cut sparsifier. Given any $\ell_1$ semi-metric $D$ on $K$, we can solve the Cut-Cut Linear Program given above. Notice that cut $(U, V-U)$ that is assigned positive weight in an optimal solution must be the minimum cut separating $U \cap K = A$ from $K - A = (V - U) \cap K$ in $G$. If not, we could replace this cut $(U, V-U)$ with the minimum cut separating $A$ from $K - A$ without affecting the feasibility and simultaneously reducing the cost of the solution. So for all $U$ for which $\delta(U) > 0$, $h(U) = h_K(U \cap K)$. 

Consider then the cost of the semi-metric $D$ against the cut-sparsifier $H$ which is defined to be $\sum_{(a, b)} c_H(a,b) D(a,b) = \sum_f \gamma(f) \sum_{(a,b)} c_f(a,b) D(a,b)$ which is just the average cost of $D$ against $G_f$ where $f$ is sampled from the distribution $\gamma$. The Cut-Cut Linear Program gives a decomposition of $D$ into a weighted sum of cut-metrics - i.e. $D(a,b) = \sum_{U} \delta(U) \Delta_U(a,b)$. Also, the cost of $D$ against $H$ is linear in $D$ so this implies that  $$\sum_{(a, b)} c_H(a,b) D(a,b) = \sum_{(a,b)} \sum_{U}  c_H(a,b) \delta(U) \Delta_U(a,b) = \sum_{(a,b)} c_H(a,b) \delta(U) h'(U \cap K)$$
In the last line, we use $\sum_{(a,b)} c_H(a,b) \Delta_U(a,b) = h'(U \cap K)$. Then
$$\sum_{(a, b)} c_H(a,b) D(a,b) \leq \sum_{U} \delta(U) \alpha h_K(U \cap K) = \alpha OPT_{cc}(G, K, D)$$
In the inequality, we have used the fact that $H$ is an $\alpha$-quality cut-sparsifier, and in the last line we have used that $\delta(U) > 0$ implies that $h(U) = h_K(U \cap K)$. This completes the proof because the average cost of $D$ against $G_f$ where $f$ is sampled from $\gamma$ is at most $\alpha OPT_{cc}(G, K, D)$, so there must be some $f$ s.t. the cost against $D$ is at most $\alpha OPT_{cc}(G, K, D)$.
\end{proof}

We will use this strong duality between the Cut-Cut Relaxation and the Contraction Quality to show that for the graph $G$ given above, no distribution on $0$-extensions gives better than an $\Omega(\sqrt{\log k})$ quality cut-sparsifier, and all we need to accomplish this is to demonstrate an integrality gap on the example for the Cut-Cut Relaxation. 

Let's repeat the construction of $G$ here. Let $Y$ be the hypercube of size $2^d$ for $d = \log k$. Then for every node $y_s \in Y$ (i.e. $s \in \{0, 1\}^d$), we add a terminal $z_s$ and connect the terminal $z_s$ to $y_s$ using an edge of capacity $\sqrt{d}$. All the edges in the hypercube are given capacity $1$. 

Then consider the distance assignment to the edges: Each edge connecting a terminal to a node in the hypercube - i.e. an edge of the form $(z_s, y_s)$ is assigned distance $\sqrt{d}$ and every other edge in the graph is assigned distance $1$. Then let $\sigma$ be the shortest path metric on $V$ given these edge distances. 

\begin{claim}
$\sigma$ is an $\ell_1$ semi-metric on $V$, and in fact there is a weighted combination of cuts s.t. $\sigma(u,v) = \sum_{U} \delta(U) \Delta_U(u,v)$ and $\sum_U \delta(U) h(U) = O(k d)$
\end{claim}

\begin{proof}
We can take $\delta(U) = 1$ for any cut $(U, V-U)$ s.t. $U = \{z_s \cup y_s | s_i = 1\}$ - i.e. $U$ is the axis-cut corresponding to the $i^{th}$ bit. We also take $\delta(U) = \sqrt{d}$ for each $U = \{z_s\}$. This set of weights will achieve $\sigma(u,v) = \sum_{U} \delta(U) \Delta_U(u,v)$, and also there are $d$ axis cuts each of which has capacity $h(U) = \frac{k}{2}$ and there are $k$ singleton cuts of weight $\sqrt{d}$ and capacity $\sqrt{d}$ so the total cost is $O(k d)$.

\end{proof}

Yet if we take $D$ equal to the restriction of $\sigma$ on $K$, then $OPT(G, K, D) = \Omega(k d^{3/2})$:

\begin{lemma}
$OPT(G, K, D) = \Omega(k d^{3/2})$
\end{lemma}

\begin{proof}
Consider any $0$-extension $f$. And we can define the weight of any terminal $a$ as $weight_f(a) = |f^{-1}(a) | = |\{v | f(v) = a\}|$. Then $\sum_a weight_f(a) = n$ because each node in $V$ is assigned to some terminal. We can define a terminal as heavy with respect to $f$ if $weight_f(a) \geq \sqrt{k}$ and light otherwise. Obviously, $\sum_a weight_f(a) = \sum_{a \mbox{ s.t. } a \mbox{ is light}} weight_f(a) +  \sum_{a \mbox{ s.t. } a \mbox{ is heavy}} weight_f(a)$ so the sum of the sizes of either all heavy terminals or of all light terminals is at least $\frac{n}{2} = \Omega(k)$. 

Suppose that $ \sum_{a \mbox{ s.t. } a \mbox{ is light}} weight_f(a) = \Omega(k)$. For any pair of terminals $a, b$, $D(a,b) \geq \sqrt{d}$. Also for any light terminal $a$, $f^{-1}(a) - \{a\}$ is a subset of the Hypercube of at most $\sqrt{k}$ nodes, and the small-set expansion of the Hypercube implies that the number of edges out of this set is at least $\Omega(weight_f(a) \log k) = \Omega(weight_f(a) d)$. Each such edge pays at least $\sqrt{d}$ cost, because $D(a,b) \geq \sqrt{d}$ for all pairs of terminals. So this implies that the total cost of the $0$-extension $f$ is at least $\sum_{a \mbox{ s.t. } a \mbox{ is light}} \Omega(weight_f(a) d^{3/2})$. 

Suppose that $ \sum_{a \mbox{ s.t. } a \mbox{ is heavy}} weight_f(a) = \Omega(k)$. Consider any heavy terminal $z_t$, and consider any $y_s \in f^{-1}(z_t)$ and $t \neq s$. Then the edge $(y_s, z_s)$ is capacity $\sqrt{d}$ and pays a total distance of $D(z_t, z_s) \geq \sigma(y_t, y_s)$. Consider any set $U$ of $\sqrt{k}$ nodes in the Hypercube. If we attempt to pack these nodes so as to minimize $\sum_{y_s \in U} \sigma(y_s, y_t)$ for some fixed node $y_t$, then the packing that minimizes the quantity is an appropriately sized Hamming ball centered at $y_t$. In a Hamming ball centered at the node $y_t$ of at least $\sqrt{k}$ total nodes, the average distance from $y_t$ is $\Omega(\log k) = \Omega(d)$, and so this implies that $\sum_{y_s \in f^{-1}(z_t)} D(z_t, z_s) \geq \sum_{y_s \in f^{-1}(z_t)} D(y_t, y_s) \geq \Omega(weight_f(z_t) d)$. Each such edge has capacity $\sqrt{d}$ so the total cost of the $0$-extension $f$ is at least $\sum_{a \mbox{ s.t. } a \mbox{ is heavy}} \
 Omega(weight_f(a) d^{3/2})$ 
\end{proof}

And of course using our strong duality result, this integrality gap implies that any cut-sparsifier that results from a distribution on $0$-extensions has quality at least $\Omega(\sqrt{\log k})$, and this matches the current best lower bound on the integrality gap of the Semi-Metric Relaxation for $0$-extension, so in principle this could be the best lower bound we could hope for (if the integrality gap of the Semi-Metric Relaxation is in fact $O(\sqrt{\log k})$ then there are always cut-sparsifiers that results from a distribution on $0$-extensions that are quality at most $O(\sqrt{\log k})$). 

\subsection{Lower bounds for Arbitrary Cut sparsifiers} 

We will in fact use the above example to give a lower bound on the quality of \emph{any} cut-sparisifer. We will show that for the above graph, no cut-sparsifier achieves quality better than $\Omega(\log^{1/4} k)$, and this gives an exponential improvement over the previous lower bound on the quality of flow-sparsifiers (which is even a stronger requirement for sparsifiers, and hence a weaker lower bound). 

The particular example $G$ that we gave above has many symmetries, and we can use these symmetries to justify considering only symmetric cut-sparsifiers. The fact that these cut-sparsifiers can be assumed without loss of generality to have nice symmetry properties, translates to that any such cut-sparsifier $H$ is characterized by a much smaller set of variables rather than one variable for every pair of terminals. In fact, we will be able to reduce the number of variables from $k \choose 2$ to $\log k$. This in turn will allow us to consider a much smaller family of cuts in $G$ in order to derive that the system is infeasible. In fact, we will only consider sub-cube cuts (cuts in which $U = \{z_s \cup y_s | s = [0, 0, 0, .... 0, *, *, ...,*]\}$) and the Hamming ball $U = \{ z_s \cup y_s | d(y_s, y_0) \leq \frac{d}{2}\}$. 

\begin{definition}
The operation $J_s$ for some $s \in \{0, 1\}^d$ which is defined as $J_s(y_t) = y_{t + s \mod 2}$ and $J_s(z_t) = z_{t + s \mod 2}$. Also let $J_s(U) = \cup_{u \in U} J_s(u)$.
\end{definition}

\begin{definition}
For any permutation $\pi: [d] \rightarrow [d]$, $\pi(s) = [s_{\pi(1)}, s_{\pi(2)}, ... s_{\pi(d)}]$. Then the operation $J_\pi$ for any permutation $\pi$ is defined at $J_\pi(y_t) = y_{\pi(t)}$ and $T_{\pi}(z_t) = z_{\pi(t)}$. Also let $J_\pi(U) = \cup_{u \in U} T_{\pi}(u)$.
\end{definition}

\begin{claim}
 For any subset $U \subset V$ and any $s \in \{0, 1\}^d$, $h(U) = h(J_s(U))$.
\end{claim}

\begin{claim}
 For any subset $U \subset V$ and any permutation $\pi: [d] \rightarrow [d]$, $h(U) = h(J_\pi(U))$.
\end{claim}

Both of these operations are automorphisms of the weighted graph $G$ and also send the set $K$ to $K$.

\begin{lemma}~\label{lemma:auto}
If there is a cut-sparsifier $H$ for $G$ which has quality $\alpha$, then there is a cut-sparsifier $H'$ which has quality at most $\alpha$ and is invariant under the automorphisms of the weighted graph $G$ that send $K$ to $K$.
\end{lemma}

\begin{proof}
Given the cut-sparsifier $H$, we can apply an automorphism $J$ to $G$, and because $h(U) = h(J(U))$, this implies that $h_K(A) = \min_{U \mbox{ s.t. } U \cap K = A} h(U) = \min_{U \mbox{ s.t. } U \cap K = A} h(J(U))$. Also $J(U \cap K) = J(U) \cap J(K) = J(U) \cap K$ so we can re-write this last line as
$$\min_{U \mbox{ s.t. } U \cap K = A} h(J(U)) = \min_{U' \mbox{ s.t. } J(U') \cap K = J(A)}  h(J(U'))$$
And if we set $U' = J^{-1}(U)$ then this last line becomes equivalent to
$$\min_{U' \mbox{ s.t. } J(U') \cap K = J(A)}  h(J(U'))= \min_{U \mbox{ s.t. } U \cap K = J(A)}  h(U)= h_K(J(A))$$
So the result is that $h_K(A) = h_K(J(A))$ and this implies that if we do not re-label $H$ according to $J$, but we do re-label $G$, then for any subset $A$, we are checking whether the minimum cut in $G$ re-labeled according to $J$, that separates $A$ from $K -A$ is close to the cut in $H$ that separates $A$ from $K-A$. The minimum cut in the re-labeled $G$ that separates $A$ from $K -A$, is just the minimum cut in $G$ that separates $J^{-1}(A)$ from $K - J^{-1}(A)$ (because the set $J^{-1}(A)$ is the set that is mapped to $A$ under $J$). So $H$ is an $\alpha$-quality cut-sparsifier for the re-labeled $G$ iff for all $A$:
$$h_K(A) = h_K(J^{-1}(A)) \leq h'(A) \leq \alpha h_K(J^{-1}(A)) = \alpha h_K(A)$$
which is of course true because $H$ is an $\alpha$-quality cut-sparsifier for $G$. 

So alternatively, we could have applied the automorphism $J^{-1}$ to $H$ and not re-labeled $G$, and this resulting graph $H_{J^{-1}}$ would also be an $\alpha$-quality cut-sparsifier for $G$. Also, since the set of $\alpha$-quality cut-sparsifiers is convex (it is defined by a system of inequalities), we can find a cut-sparsifier $H'$ that has quality at most $\alpha$ and is a fixed point of the group of automorphisms, and hence invariant under the automorphisms of $G$ as desired.
\end{proof}

\begin{corollary}\label{cor:auto}
If $\alpha$ is the best quality cut-sparsifier for the above graph $G$, then there is an $\alpha$ quality cut-sparsifier $H$ in which the capacity between two terminals $z_s$ and $z_t$ is only dependent on the Hamming distance $Hamm(s, t)$.
\end{corollary}

\begin{proof}
Given any quadruple $z_s, z_t$ and $z_{s'}, z_{t'}$ s.t. $Hamm(s, t) = Hamm(s', t')$, there is a concatenation of operations from $J_s$, $J_\pi$ that sends $s$ to $s'$ and $t$ to $t'$. This concatenation of operations $J$ is in the group of automorphisms that send $K$ to $K$, and hence we can assume that $H$ is invariant under this operation which implies that $c_H(s,t) = c_H(s',t')$.
\end{proof}

One can regard any cut-sparsifier (not just ones that result from contractions) as a set of $k \choose 2$ variables, one for the capacity of each edge in $H$. Then the constraints that $H$ be an $\alpha$-quality cut-sparsifier are just a system of inequalities, one for each subset $A \subset K$ that enforces that the cut in $H$ is at least as large as the minimum cut in $G$ (i.e. $h'(A) \geq h_K(A)$) and one enforcing that the cut is not too large (i.e. $h'(A) \leq \alpha h_K(A)$). Then in general, one can derive lower bounds on the quality of cut-sparsifiers by showing that if $\alpha$ is not large enough, then this system of inequalities is infeasible meaning that there is not cut-sparsifier achieving quality $\alpha$. Unlike the above argument, this form of a lower bound is much stronger and does not assume anything about how the cut-sparsifier is generated. 

\begin{theorem1}{Theorem}{thm:lower}
For $\alpha = \Omega(\log^{1/4} k)$, there is no cut-sparsifier $H$ for $G$ which has quality at most $\alpha$.
\end{theorem1}

\vspace{0.5pc}

\begin{proofsketch}
Assume that there is a cut-sparsifier $H'$ of quality at most $\alpha$. Then using the above corollary, there is a cut-sparsifier $H$ of quality at most $\alpha$ in which the weight from $a$ to $b$ is only a function of $Hamm(a,b)$. Then for each $i \in [d]$, we can define a variable $w_i$ as the total weight of edges incident to any terminal of length $i$. I.e. $w_i = \sum_{b \mbox{ s.t. } Hamm(a,b) = i} c_H(a,b)$.

For simplicity, here we will assume that all cuts in the sparsifier $H$ are at most the cost of the corresponding minimum cut in $G$ and at least $\frac{1}{\alpha}$ times the corresponding minimum cut. This of course is an identical set of constraints that we get from dividing the standard definition that we use in this paper for $\alpha$-quality cut-sparsifiers by $\alpha$. 

We need to derive a contradiction from the system of inequalities that characterize the set of $\alpha$-quality cut sparsifiers for $G$. As we noted, we will consider only the sub-cube cuts (cuts in which $U = \{z_s \cup y_s | s = [0, 0, 0, .... 0, *, *, ... *]\}$) and the Hamming ball $U = \{ z_s \cup y_s | d(y_s, y_0) \leq \frac{d}{2}\}$, which we refer to as the Majority Cut.

Consider the Majority Cut: There are $\Theta(k)$ terminals on each side of the cut, and most terminals have Hamming weight close to $\frac{d}{2}$. In fact, we can sort the terminals by Hamming weight and each weight level around Hamming weight $\frac{d}{2}$ has roughly a $\Theta(\frac{1}{\sqrt{d}})$ fraction of the terminals. Any terminal of Hamming weight $\frac{d}{2} - \sqrt{i}$ has roughly a constant fraction of their weight $w_i$ crossing the cut in $H$, because choosing a random terminal Hamming distance $i$ from any such terminal corresponds to flipping $i$ coordinates at random, and throughout this process there are almost an equal number of $1$s and $0$s so this process is well-approximated by a random walk starting at $\sqrt{i}$ on the integers, which equally likely moves forwards and backwards at each step for $i$ total steps, and asking the probability that the walk ends at a negative integer. 

In particular, for any terminal of Hamming weight $\frac{d}{2} - t$, the fraction of the weight $w_i$ that crosses the Majority Cut is $O(exp\{-\frac{t^2}{i})$. So the total weight of length $i$ edges (i.e. edges connecting two terminals at Hamming distance $i$) cut by the Majority Cut is $O(w_i |\{z_s | Hamm(s, 0) \geq \frac{d}{2} - \sqrt{i}\}|) = O(w_i \sqrt{i / d})k$ because each weight close to the boundary of the Majority cut contains roughly a $\Theta(\frac{1}{\sqrt{d}})$ fraction of the terminals. 
So the total weight of edges crossing the Majority Cut in $H$ is
$O(k \sum_{i = 1}^d w_i \sqrt{i / d})$

And the total weight crossing the minimum cut in $G$ separating $A = \{ z_s  | d(y_s, y_0) \leq \frac{d}{2}\}$ from $K - A$ is $\Theta(k \sqrt{d})$. And because the cuts in $H$ are at least $\frac{1}{\alpha}$ times the corresponding minimum cut in $G$, this implies
$\sum_{i = 1}^d w_i \sqrt{i / d} \geq \Omega(\frac{\sqrt{d}}{\alpha})$

Next, we consider the set of sub-cube cuts. For $j \in [d]$, let $A_j = \{z_s | s_1 = 0, s_2 = 0, .. s_j = 0\}$. Then the minimum cut in $G$ separating $A_j$ from $K - A_j$ is $\Theta(|A_j| \min(j, \sqrt{d}))$, because each node in the Hypercube which has the first $j$ coordinates as zero has $j$ edges out of the sub-cube, and when $j > \sqrt{d}$, we would instead choose cutting each terminal $z_s \in A_j$ from the graph directly by cutting the edge $(y_s, z_s)$. 

Also, for any terminal in $A_j$, the fraction of length $i$ edges that cross the cut is approximately $1 - (1 - \frac{j}{d})^i = \Theta(\min(\frac{ij}{d}, 1))$. So the constraints that each cut in $H$ be at most the corresponding minimum cut in $G$ give the inequalities
$ \sum_{i = 1}^d \min(\frac{i j}{d}, 1) w_i \leq O(\min(j, \sqrt{d}))$

We refer to the above constraint as $B_j$. Multiply each $B_j$ constraint by $\frac{1}{j^{3/2}}$ and adding up the constraints yields a linear combination of the variables $w_i$ on the left-hand side. The coefficient of any $w_i$ is $$\sum_{j=1}^{d-1}  \frac{\min(\frac{i j}{d}, 1)}{j^{3/2}} \geq \sum_{j=1}^{d/i} \frac{\frac{i j}{d}}{j^{3/2}}$$

And using the Integration Rule this is $\Omega(\sqrt{\frac{i}{d}})$. 

This implies that the coefficients of the constraint $B$ resulting from adding up $\frac{1}{j^{3/2}}$ times each $B_j$ for each $w_i$ are at least as a constant times the coefficient of $w_i$ in the Majority Cut Inequality. So we get

$$\sum_{j=1}^{d-1} \frac{1}{j^{3/2}}\min(j, \sqrt{d}) \geq \Omega \Big (\sum_{j=1}^{d-1}  \frac{1}{j^{3/2}} \sum_{i = 1}^d \min(\frac{i j}{d}, 1) w_i \Big) \geq \Omega \Big(\sum_{i = 1}^d w_i \sqrt{\frac{i}{d}} \Big) \geq \Omega \Big(\frac{\sqrt{d}}{\alpha} \Big)$$

And we can evaluate the constant $\sum_{j=1}^{d-1} j^{-3/2} \min(j, \sqrt{d}) =\sum_{j=1}^{\sqrt{d}} j^{-1/2} + \sqrt{d} \sum_{j=\sqrt{d}+1}^{d-1} j^{-3/2}$
using the Integration Rule, this evaluates to $O(d^{1/4})$. This implies $O(d^{1/4}) \geq \frac{\sqrt{d}}{\alpha}$ and in particular this implies $\alpha \geq \Omega(d^{1/4})$. So the quality of the best cut-sparsifier for $H$ is at least $\Omega(\log^{1/4} k)$. 
\end{proofsketch}

We note that this is the first super-constant lower bound on the quality of cut-sparsifiers. Recent work gives a super-constant lower bound on the quality of flow-sparsifiers in an infinite family of expander-like graphs. However, for this family there are constant-quality cut-sparsifiers. In fact, lower bounds for cut-sparsifiers imply lower bounds for flow-sparsifiers, so we are able to improve the lower bound of $\Omega(\log \log k)$ in the previous work for flow-sparsifiers by an exponential factor to $\Omega(\log^{1/4} k)$, and this is the first lower bound that is tight to within a polynomial factor of the current best upper bound of $O(\frac{\log k}{\log \log k})$. 

This bound is not as good as the lower bound we obtained earlier in the restricted case in which the cut-sparsifier is generated as a convex combination of $0$-extension graphs $G_f$. As we will demonstrate, there are actually cut-sparsifiers that achieve quality $o(\sqrt{\log k})$ for $G$, and so in general restricting to convex combinations of $0$-extensions is sub-optimal, and we leave open the possibility that the ideas in this improved bound may result in better constructions of cut (or flow)-sparsifiers that are able to beat the current best upper bound on the integrality gap of the $0$-extension linear program.


\section{Noise Sensitive Cut-Sparsifiers}

In Appendix~\ref{sec:aha}, we give a brief introduction to the harmonic analysis of Boolean functions, along with formal statements that we will use in the proof of our main theorem in this section. 

\subsection{A Candidate Cut-Sparsifier}

Here we give a cut-sparsifier $H$ which will achieve quality $o(\sqrt{\log k})$ for the graph $G$ given in Section~\ref{sec:lb}, which is asymptotically better than the best cut-sparsifier that can be generated from contractions. 

As we noted, we can assume that the weight assigned between a pair of terminals in $H$, $c_H(a,b)$ is only a function of the Hamming distance from $a$ to $b$. In $G$, the minimum cut separating any singleton terminal $\{z_s\}$ from $K - \{z_s\}$ is just the cut that deletes the edge $(z_s, y_s)$. So the capacity of this cut is $\sqrt{d}$. We want a good cut-sparsifier to approximately preserve this cut, so the total capacity incident to any terminal in $H$ will also be $\sqrt{d}$ - i.e. $c'(\{z_s\}) = \sqrt{d}$. 

We distribute this capacity among the other terminals as follows: We sample $t \sim_\rho s$, and allocate an infinitesimal fraction of the total weight $\sqrt{d}$ to the edge $(z_s, z_t)$. Equivalently, the capacity of the edge connecting $z_s$ and $z_t$ is just $Pr_{u \sim_\rho t}[u = s] \sqrt{d}$. We choose $\rho = 1 - \frac{1}{\sqrt{d}}$. This choice of $\rho$ corresponds to flipping each bit in $t$ with probability $\Theta(\frac{1}{\sqrt{d}})$ when generating $u$ from $t$. We prove that the graph $H$ has cuts at most the corresponding minimum-cut in $G$. 

This cut-sparsifier $H$ has cuts at most the corresponding minimum-cut in $G$. In fact, a stronger statement is true: $\vec{H}$ can be routed as a flow in $G$ with congestion $O(1)$. Consider the following explicit routing scheme for $\vec{H}$: Route the $\sqrt{d}$ total flow in $\vec{H}$ out of $z_s$ to the node $y_s$ in $G$. Now we need to route these flows through the Hypercube in a way that does not incur too much congestion on any edge. Our routing scheme for routing the edge from $z_s$ to $z_t$ in $\vec{H}$ from $y_s$ to $y_t$ will be symmetric with respect to the edges in the Hypercube: choose a random permutation of the bits $\pi : [d] \rightarrow [d]$, and given $u \sim_\rho t$, fix each bit in the order defined by $\pi$. So consider $i_1 = \pi(1)$. If $t_{i_1} \neq u_{i_1}$, and the flow is currently at the node $x$, then flip the $i_1^{th}$ bit of $x$, and continue for $i_2 = \pi(2)$, $i_3, ... i_d = \pi(d)$. 

Each permutation $\pi$ defines a routing scheme, and we can average over all permutations $\pi$ and this results in a routing scheme that routes $\vec{H}$ in $G$. 

\begin{claim}
This routing scheme is symmetric with respect to the automorphisms $J_s$ and $J_\pi$ of $G$ defined above. 
\end{claim}

\begin{corollary}
The congestion on any edge in the Hypercube incurred by this routing scheme is the same. 
\end{corollary}

\begin{lemma}
The above routing scheme will achieve congestion at most $O(1)$ for routing $\vec{H}$ in $G$.
\end{lemma}

\begin{proof}
Since the congestion of any edge in the Hypercube under this routing scheme is the same, we can calculate the worst case congestion on any edge by calculating the average congestion. Using a symmetry argument, we can consider any fixed terminal $z_s$ and calculate the expected increase in average congestion when sampling a random permutation $\pi: [d] \rightarrow [d]$ and routing all the edges out of $z_s$ in $H$ using $\pi$. This expected value will be $k$ times the average congestion, and hence the worst-case congestion of routing $\vec{H}$ in $G$ according to the above routing scheme. 

As we noted above, we can define $H$ equivalently as arising from the random process of sampling $u \sim_\rho t$, and routing an infinitesimal fraction of the $\sqrt{d}$ total capacity out of $z_t$ to $z_u$, and repeating until all of the $\sqrt{d}$ capacity is allocated. We can then calculate the the expected increase in average congestion (under a random permutation $\pi$) caused by routing the edges out of $z_s$ as the expected increase in average congestion divided by the total fraction of the $\sqrt{d}$ capacity allocated when we choose the target $u$ from $u \sim_\rho t$. In particular, if we allocated a $\Delta $ fraction of the $\sqrt{d}$ capacity, the expected increase in total congestion is just the total capacity that we route multiplied by the length of the path. Of course, the length of this path is just the number of bits in which $u$ and $t$ differ, which in expectation is $\Theta(\sqrt{d})$ by our choice of $\rho$. 

So in this procedure, we allocate $\Delta  \sqrt{d}$  total capacity, and the expected increase in total congestion is the total capacity routed $\Delta \sqrt{d}$ times the expected path length $\Theta(\sqrt{d})$. We repeat this procedure $\frac{1}{\Delta}$ times, and so the expected increase in total congestion caused by routing the edges out of $z_t$ in $G$ is $\Theta(d)$. If we perform this procedure for each terminal, the resulting total congestion is $\Theta(kd)$, and because there are $\frac{kd}{2}$ edges in the Hypercube, the average congestion is $\Theta(1)$ which implies that the worst-case congestion on any edge in the Hypercube is also $O(1)$, as desired. Also, the congestion on any edge $(z_s, y_s)$ is $1$ because there is a total of $\sqrt{d}$ capacity out of $z_s$ in $H$, and this is the only flow routed on this edge, which has capacity $\sqrt{d}$ in $G$ by construction. So the worst-case congestion on any edge in the above routing scheme is $O(1)$. 
\end{proof}

\vspace{0.5pc}

\begin{corollary}\label{cor:cutlow}
For any $A \subset K$, $h'(A) \leq O(1) h_K(A)$.
\end{corollary}

\vspace{0.5pc}

\begin{proof}
Consider any set $A \subset K$. Let $U$ be the minimum cut in $G$ separating $A$ from $K - A$. Then the total flow routed from $A$ to $K - A$ in $\vec{H}$ is just $h'(A)$, and if this flow can be routed in $G$ with congestion $O(1)$, this implies that the total capacity crossing the cut from $U$ to $V-U$ is at least $\Omega(1) h'(A)$. And of course the total capacity crossing the cut from $U$ to $V-U$ is just $h_K(A)$ by the definition of $U$, which implies the corollary.
\end{proof}

So we know that the cuts in $H$ are never too much larger than the corresponding minimum cut in $G$, and all that remains to show that the quality of $H$ is $o(\sqrt{\log k})$ is to show that the cuts in $H$ are never too small. We conjecture that the quality of $H$ is actually $\Theta(\log^{1/4} k)$, and this seems natural since the quality of $H$ just restricted to the Majority Cut and the sub-cube cuts is actually $\Theta(\log^{1/4} k)$, and often the Boolean functions corresponding to these cuts serve as extremal examples in the harmonic analysis of Boolean functions. In fact, our lower bound on the quality of any cut-sparsifier for $G$ is based only on analyzing these cuts so in a sense, our lower bound is tight given the choice of cuts in $G$ that we used to derive infeasibility in the system of equalities characterizing $\alpha$-quality cut-sparsifiers. 

\subsection{A Fourier Theoretic Characterization of Cuts in $H$}

Here we give a simple formula for the size of a cut in $H$, given the Fourier representation of the cut. So here we consider cuts $A \subset K$ to be Boolean functions of the form $f_A: \{-1, +1\}^d \rightarrow \{-1, +1\}$ s.t. $f_A(s) = +1$ iff $z_s \in A$. 

\begin{lemma}
$h'(A) = k \frac{\sqrt{d}}{2} \frac{1 - NS_{\rho}[f_A(x)]}{2}$
\end{lemma}

\begin{proof}
We can again use the infinitesimal characterization for $H$, in which we choose $u \sim_\rho t$ and allocate $\Delta$ units of capacity from $z_s $ to $z_t$ and repeat until all $\sqrt{d}$ units of capacity are spent. 

If we instead choose $z_s$ uniformly at random, and then choose  $u \sim_\rho t$ and allocate $\Delta$ units of capacity from $z_s $ to $z_t$, and repeat this procedure until all $k \frac{\sqrt{d}}{2}$ units of capacity are spent, then at each step the expected contribution to the cut is exactly $\Delta \frac{1 - NS_{\rho}[f_A(x)]}{2}$ because $\frac{1 - NS_{\rho}[f_A(x)]}{2}$ is exactly the probability that if we choose $t$ uniformly at random, and $u \sim_\rho t$ that $f_A(u) \neq f_A(t)$ which means that this edge contributes to the cut. We repeat this procedure $\frac{k \sqrt{d}}{2 \Delta}$ times, so this implies the lemma. 
\end{proof}

\vspace{0.5pc}

\begin{lemma}\label{lemma:fcut}
$h'(A) = \Theta \Big (k \sum_S \hat{f}_S^2 \min(|S|, \sqrt{d}) \Big )$
\end{lemma}

\vspace{0.5pc}

\begin{proof}
Using the setting $\rho = 1 - \frac{1}{\sqrt{d}}$, we can compute $h'(A)$ using the above lemma:

$$h'(A) = k \frac{\sqrt{d}}{4} (1 - NS_{\rho}[f_A(x)])$$

And using Parseval's Theorem, $\sum_S \hat{f}_S^2 = ||f||_2 = 1$, so we can replace $1$ with $\sum_S \hat{f}_S^2$ in the above equation and this implies

$$h'(A) = k \frac{\sqrt{d}}{4} \sum_S \hat{f}_S^2 (1 - (1 - \frac{1}{\sqrt{d}})^{|S|})$$

Consider the term $(1 - (1 - \frac{1}{\sqrt{d}})^{|S|})$. For $|S| \leq \sqrt{d}$, this term is $\Theta(\frac{|S|}{\sqrt{d}})$, and if $|S| \geq \sqrt{d}$, this term is $\Theta(1)$. So this implies

$$h'(A) = \Theta \Big (k \sum_S \hat{f}_S^2 \min(|S|, \sqrt{d}) \Big )$$
\end{proof}

\subsection{Small Set Expansion of $H$}

The edge-isoperimetric constant of the Hypercube is $1$, but on subsets of the cube that are imbalanced, the Hypercube expands more than this. 

\begin{definition}
For a given set $A \subset \{-1, +1\}^{[d]}$, we define $bal(A) = \frac{1}{k} \min(|A|, k - |A|)$ as the balance of the set $A$. 
\end{definition}

Given any set $A  \subset \{-1, +1\}^{[d]}$ of balance $b = bal(A)$, the number of edges crossing the cut $(A, \{-1, +1\}^{[d]} - A)$ in the Hypercube is $\Omega(b k \log \frac{1}{b})$. So the Hypercube expands better on small sets, and we will prove a similar small set expansion result for the cut-sparsifier $H$. In fact, for any set $A \subset K$ (which we will associated with a subset of $\{-1, +1\}^{[d]} $ and abuse notation), $h'(A) \geq bal(A) k \Omega(\min(\log \frac{1}{bal(A)}, \sqrt{d}))$. We will prove this result using the Hypercontractive Inequality. 

\begin{lemma}\label{lemma:imb}
$h'(A) \geq bal(A) k \Omega(\min(\log \frac{1}{bal(A)}, \sqrt{d}))$
\end{lemma}

\vspace{0.5pc}

\begin{proof}
Assume that $|A| \leq |\{-1, +1\}^{[d]} - A|$ without loss of generality. Throughout just this proof, we will use the notation that $f_A : \{-1, +1\}^d \rightarrow \{0, 1\}$ and $f_A(s) = 1$ iff $s \in A$. Also we will denote $b = bal(A)$. 

Let $\gamma << 1$ be chose later. Then we will invoke the Hypercontractive inequality with $q =2$, $p = 2 - \gamma$, and $\rho = \sqrt{\frac{p-1}{q-1}} = \sqrt{1 - \gamma}$. Then

$$||f||_p = E_x[f(x)^p]^{1/p} = b^{1/p} \approx b^{1/2 (1 + \gamma/2)}$$

Also $||T_{\rho}(f(x))||_q = ||T_{\rho}(f(x))||_2 = \sqrt{\sum_S \rho^{2 |S|} \hat{f}_S^2}$. So the Hypercontractive Inequality implies

$$\sum_S \rho^{2 |S|} \hat{f}_S^2 \leq b^{1 + \gamma/2} = b e^{-\frac{\gamma}{2} \ln\frac{1}{b}}$$

And $\rho^{2 |S|} = (1 - \gamma)^{|S|}$. Using Parseval's Theorem, $\sum_S \hat{f}_S^2 = ||f||_2^2  = b$, and so we can re-write the above inequality as

$$b - b e^{\frac{-\gamma}{2} \ln\frac{1}{b}} \leq b - \sum_S (1 - \gamma)^{ |S|} \hat{f}_S^2 \leq b^{1 - \gamma/2}  = \sum_S \hat{f}_S^2 (1 - (1 - \gamma)^{ |S|}) = \Theta \Big (\sum_S \hat{f}_S^2 \gamma \min(|S|, \frac{1}{\gamma}) \Big )$$

This implies

$$\frac{b}{\gamma} (1 - e^{-\frac{\gamma}{2} \ln\frac{1}{b}}) \leq \Theta \Big (\sum_S \hat{f}_S^2  \min(|S|, \frac{1}{\gamma}) \Big )$$

And as long as $\frac{1}{\gamma} \leq \sqrt{d}$, 

$$\sum_S \hat{f}_S^2  \min(|S|, \frac{1}{\gamma}) \leq \frac{1}{k}O(h'(A)) = O \Big ( \sum_S \hat{f}_S^2 \min(|S|, \sqrt{d}) \Big )$$

If $\frac{\gamma}{2} \ln\frac{1}{b} \leq 1$, then $e^{-\frac{\gamma}{2} \ln\frac{1}{b}} = 1 - \Omega(\frac{\gamma}{2} \ln\frac{1}{b})$ which implies

$$\frac{b}{\gamma} (1 - e^{-\frac{\gamma}{2} \ln\frac{1}{b}}) \geq \Omega( b \ln \frac{1}{b})$$

However if $\ln \frac{1}{b} = \Omega(\sqrt{d})$, then we cannot choose $\gamma$ to be small enough (we must choose $\frac{1}{\gamma} \leq \sqrt{d}$) in order to make $\frac{\gamma}{2} \ln\frac{1}{b}$ small. 

So the only remaining case is when $\ln \frac{1}{b} = \Omega(\sqrt{d})$. Then notice that the quantity $(1 - e^{-\frac{\gamma}{2} \ln\frac{1}{b}})$ is increasing with decreasing $b$. So we can lower bound this term by substituting $b = e^{-\Theta(\sqrt{d})}$. If we choose $\gamma = \frac{1}{\sqrt{d}}$ then this implies

$$\frac{1}{\gamma}(1 - e^{-\frac{\gamma}{2} \ln\frac{1}{b}}) = \Omega(\sqrt{d})$$

And this in turn implies that

$$\frac{b}{\gamma}(1 - e^{-\frac{\gamma}{2} \ln\frac{1}{b}}) = \Omega(b \sqrt{d})$$
which yields $h'(A) \geq \Omega(bk \sqrt{d})$. So in either case, $h'(A)$ is lower bounded by either $\Omega(bk \sqrt{d}) $ or $\Omega(bk \ln \frac{1}{b})$, as desired. 

\end{proof}

\subsection{Interpolating Between Cuts via Bourgain's Junta Theorem}

In this section, we show that the quality of the cut-sparsifier $H$ is $o(\sqrt{\log k})$, thus beating how well the best distribution on $0$-extensions can approximate cuts in $G$ by a super-constant factor. 

We will first give an outline of how we intend to combine Bourgain's Junta Theorem, and the small set expansion of $H$ in order to yield this result. In a previous section, we gave a Fourier theoretic characterization of the cut function of $H$. We will consider an arbitrary cut $A \subset K$ and assume for simplicity that $|A| \leq |K-A|$. If the Boolean function $f_A$ that corresponds to this cut has significant mass at the tail of the spectrum, this will imply (by our Fourier theoretic characterization of the cut function) that the capacity of the corresponding cut in $H$ is $\omega(k)$. Every cut in $G$ has capacity at most $O(k \sqrt{d})$ because we can just cut every edge $(z_s, y_s)$ for each terminal $z_s \in A$, and each such edge has capacity $\sqrt{d}$. Then in this case, the ratio of the minimum cut in $G$ to the corresponding cut in $H$ is $o(\sqrt{d})$. 

But if the tail of the Fourier spectrum of $f_A$ is not significant, and applying Bourgain's Junta Theorem implies that the function $f_A$ is close to a junta. Any junta will have a small cut in $G$ (we can take axis cuts corresponding to each variable in the junta) and so for any function that is different from a junta on a vanishing fraction of the inputs, we will be able to construct a cut in $G$ (not necessarily minimum) that has capacity $o(k \sqrt{d})$. On all balanced cuts (i.e. $|A| = \Theta(k)$), the capacity of the cut in $H$ will be $\Omega(k)$, so again in this case the ratio of the minimum cut in $G$ to the corresponding cut in $H$ is $o(\sqrt{d})$. 

So the only remaining case is when $|A| = o(k)$, and from the small set expansion of $H$ the capacity of the cut in $H$ is $\omega(|A|)$ because the cut is imbalanced. Yet the minimum cut in $G$ is again at most $|A| \sqrt{d}$, so in this case as well the ratio of the minimum cut in $G$ to the corresponding cut in $H$ is $o(\sqrt{d})$. 

\begin{theorem1}{Theorem}{thm:sep}
There is an infinite family of graphs for which the quality of the best
cut-sparsifier is $\Omega(\frac{\log^2 \log \log k}{\log\log\log\log k})$ better
than the best that a distribution on $0$-extensions can achieve. 
\end{theorem1}

We repeat Bourgain's Junta Theorem:

\begin{theorem}
[Bourgain] \cite{Bou}, \cite{KN}
Let $f \{-1, +1\}^d \rightarrow \{-1, +1\}$ be a Boolean function. Then fix any $\epsilon, \delta \in (0, 1/10)$.  Suppose that
$$\sum_S (1-\epsilon)^{|S|} \hat{f}_S^2 \geq 1 - \delta$$
then for every $\beta > 0$, $f$ is a
$$\Big (2^{c \sqrt{\log 1/\delta \log \log 1/\epsilon}}\Big ( \frac{\delta}{\sqrt{\epsilon}} + 4^{1/\epsilon} \sqrt{\beta}\Big), \frac{1}{\epsilon \beta} \Big ) \mbox{-junta}$$
\end{theorem}

We will choose:

$\frac{1}{\epsilon} = \frac{1}{32} \log d$

$\frac{1}{\beta} = d^{1/4}$

$\frac{1}{\delta'} = \log^{2/3} d$

And also let $b = bal(A) = \frac{|A|}{k}$, and remember for simplicity we have
assumed that $|A| \leq |K-A|$, so $b \leq \frac{1}{2}$.

$\delta = \delta'b$

\begin{lemma}
If $\sum_S (1-\epsilon)^{|S|} \hat{f}_S^2 \leq 1 - \delta$ then this implies
$\sum_S \hat{f}_S^2 \min(|S|, \sqrt{d}) \geq \Omega \Big
(\frac{\delta}{\epsilon} \Big ) = \Omega(b\log^{1/3} d)$
\end{lemma}

\begin{proof}
The condition $\sum_S (1-\epsilon)^{|S|} \hat{f}_S^2 \geq 1 - \delta$ implies $\delta \leq 1 - \sum_S (1-\epsilon)^{|S|} \hat{f}_S^2 = O(\sum_S \hat{f}_S^2 \min(|S| \epsilon, 1)) $ and rearranging terms this implies
$$\frac{\delta}{\epsilon} \leq O(\sum_S \hat{f}_S^2 \min(|S|, \frac{1}{\epsilon})) = O(\sum_S \hat{f}_S^2 \min(|S|, \sqrt{d}))$$
where the last line follows because $\frac{1}{\epsilon} = O(\log d) \leq O(\sqrt{d})$.
\end{proof}

So combining this lemma and Lemma~\ref{lemma:fcut}: if the conditions of
Bourgain's Junta Theorem are not met, then the capacity of the cut in the
sparsifier is $\Omega(kb \log^{1/3} d)$. And of course, the capacity of the
minimum cut in $G$ is at most $kb \sqrt{d}$, because for each $z_s \in A$ we
could separate $A$ from $K -A $ by cutting the edge $(z_s, y_s)$, each of which
has capacity $\sqrt{d}$. 

\begin{case}
If the conditions of Bourgain's Junta Theorem are not met, then the ratio of the
minimum cut in $G$ separating $A$ from $K -A$ to the corresponding cut in $H$ is
at most $O(\frac{\sqrt{d}}{\log^{1/3}d})$.
\end{case}

But what if the conditions of Bourgain's Junta Theorem are met? We can check
what Bourgain's Junta Theorem implies for the given choice of parameters.  We
first consider the case when $b$ is not too small. In particular, for our choice of parameters the following 3 inequalities hold:
\begin{eqnarray}
2^{c \sqrt{\log 1/\delta \log \log 1/\epsilon}} &\leq& \log^{1/24} d
\label{inequ:b-is-small-1}\\
\frac{\delta}{\sqrt{\epsilon}} &\geq& 4^{1/\epsilon}\sqrt{\beta}
\label{inequ:b-is-small-2}\\
\sqrt{d}b\log^{-1/8}d &\geq& d^{1/4}\log d \label{inequ:b-is-small-3}
\end{eqnarray}

\begin{claim}
If (\ref{inequ:b-is-small-2}) is true, $\left( \frac{\delta}{\sqrt{\epsilon}} +
4^{1/\epsilon}
\sqrt{\beta}\right) = O\left(b\log^{-1/6} d\right)$
\end{claim}

\begin{claim}
If (\ref{inequ:b-is-small-1}) and (\ref{inequ:b-is-small-2}) are true, $2^{c
\sqrt{\log 1/\delta \log \log 1/\epsilon}}\Big (
\frac{\delta}{\sqrt{\epsilon}} + 4^{1/\epsilon} \sqrt{\beta}\Big) =
O\left(b\log^{-1/8}d\right)$
\end{claim}

So when we apply Bourgain's Junta Theorem, if the conditions are met (for our
given choice of parameters), we get that $f_A$ is an
$\left(O\left(b\log^{-1/8}d\right),
O(d^{1/4} \log d)\right)$-junta. 

\begin{lemma}
If $f_A$ is a $(\nu, j)$-junta, then $h_K(A) \leq k \nu \sqrt{d} + j \frac{k}{2}$
\end{lemma}

\begin{proof}
Let $g$ be a $j$-junta s.t. $Pr_x[f_A(x) \neq g(x)] \leq \nu$. Then we can disconnect the set of nodes on the Hypercube where $g$ takes a value $+1$ from the set of nodes where $g$ takes a value $-1$ by performing an axis cut for each variable that $g$ depends on. Each such axis cut, cuts $\frac{k}{2}$ edges in the Hypercube, so the total cost of cutting these edges is $j \frac{k}{2}$ and then we can alternatively cut the edge $(z_s, y_s)$ for any $s$ s.t. $f_A(s) \neq g(s)$, and this will be a cut separating $A$ from $K -A$ and these extra edges cut are each capacity $\sqrt{d}$ and we cut at most $\nu k$ of these edges in total.
\end{proof}

So if $f_A$ is an$\left(O\left(b\log^{-1/8}d\right), O(d^{1/4} \log
d)\right)$-junta and (\ref{inequ:b-is-small-3}) holds, then $h_K(A)
\leq O\left(\frac{kb \sqrt{d}}{\log^{1/8}d }\right)$. 

\begin{case}
Suppose the conditions of Bourgain's Junta Theorem are met, and
(\ref{inequ:b-is-small-1})(\ref{inequ:b-is-small-2}) and
(\ref{inequ:b-is-small-3}) are true, then the ratio
of the minimum cut in $G$ separating $A$ from $K -A$ to the corresponding cut in
$H$ is at most $O(\frac{\sqrt{d}}{\log^{1/8}d})$.
\end{case}

\begin{proof}
Lemma~\ref{lemma:imb} also implies that the edge expansion of $H$ is
$\Omega(1)$, so given a cut $|A|$, $h'(A) \geq \Omega (|A|) = \Omega(kb)$. Yet
under the conditions of this case, the capacity of
the cut in $G$ is $O\left(\frac{kb\sqrt{d}}{\log^{1/8}k
}\right)$ and this implies the
statement.
\end{proof}

So, the only remaining case is when the conditions of Bourgain's Junta Theorem
are met at least 1 of the 3 conditions is not true. Yet we can apply Lemma
~\ref{lemma:imb} directly
to get that in this case $h'(A) = \omega(|A|)$ and of course $h_K(A) \leq
|A|\sqrt{d}$.

\begin{case}
 Suppose the conditions of Bourgain's Junta Theorem are met, and at least 1 of
the 3 inequalities is not true, then the ratio of the minimum cut in $G$
separating $A$ from
$K-A$ to the corresponding cut in $H$ is at most $O(\frac{\sqrt{d}\log\log\log
d}{\log^2\log d})$.
\end{case}

\begin{proof}
If (\ref{inequ:b-is-small-1}) is false, $\log(1/\delta') + \log(1/b) =
\log(1/\delta) > \frac{(\log
\log^{1/30}d/c)^2}{\log\log 1/\epsilon}=\Omega\left(\frac{\log^2\log
d}{\log\log\log d}\right)$. Since $1 / \delta' = O(\log\log d)$, it must be the
case that $\log(1/b) = \Omega\left(\frac{\log^2\log
d}{\log\log\log d}\right)$.

If (\ref{inequ:b-is-small-2}) is false, $b <
\frac{4^{1/\epsilon}\sqrt\beta\sqrt\epsilon}{\delta'}=O(d^{-1/8}\log^{1/6}d)$,
and $\log (1/b) = \Omega(\log d)$.

If (\ref{inequ:b-is-small-3}) is false, $b < d^{-1/4}\log^{9/8}d$ and $\log
(1/b) = \Omega(\log d)$.

The minimum of the 3 bounds is the first one. So, $\log(1/b) =
\Omega\left(\frac{\log^2\log
d}{\log\log\log d}\right)$ if at least 1 of the 3 conditions is false. Applying
Lemma \ref{lemma:imb}, we get that $h'(A)
\geq \Omega(|A|\log\frac1b)=\Omega(|A|\frac{\log^2\log d}{\log\log\log d})$. 
And yet $h_K(A) \leq |A|\sqrt{d}$, and this implies the statement. 
Combining the cases, this implies that the quality of $H$ is
$O(\frac{\sqrt{d}\log\log\log d}{\log^2 \log d})$.
\end{proof}

\begin{conjecture}
The quality of $H$ as a cut-sparsifier for $G$ is $O(d^{1/4})$
\end{conjecture}

\vspace{0.5pc}

\begin{theorem1}{Theorem}{thm:sep}
There is an infinite family of graphs for which the quality of the best
cut-sparsifier is $\Omega(\frac{\log^2 \log \log k}{\log\log\log\log k})$ better
than the best that a distribution on $0$-extensions can achieve. 
\end{theorem1}

\vspace{0.5pc}

\section{Improved Constructions via Lifting}\label{sec:alift}

In this section we give a polynomial time construction for a flow-sparsifier
that achieves quality at most the quality of the best flow-sparsifier that can
be realized as a distribution over $0$-extensions. Thus we give a construction
for flow-sparsifiers (and thus also cut-sparsifier) that achieve quality
$O(\frac{\log k}{\log \log k})$. Given that the current best upper bounds on the
quality of both flow and cut-sparsifiers are achieved as a distribution over
$0$-extensions, the constructive result we present here matches the best known
\emph{existential} bounds on the quality of cut or flow-sparsifiers. All
previous constructions \cite{M}, \cite{LM} need to sacrifice some super-constant
factor in order to actually construct cut or flow-sparsifiers. We achieve this
using a linear program that can be interpreted as a lifting of previous linear
programs used in constructive results. 

Our technique, we believe, is of independent interest: we perform a lifting on an appropriate linear program. This lifting allows us to implicitly enforce a constraint automatically that previously was difficult to enforce, and required an approximate separation oracle rather than an exact separation oracle. 

There are known ways for implicitly enforcing this constraint using an exponential number of variables, but surprisingly we are able to implicitly enforce this constraint using only polynomially many variables, after just a single lifting operation. The lifting operation that we perform is inspired by Earth-mover relaxations, and makes it a rare example of when an algorithm is actually able to use the Earth-mover constraints, as opposed to the usual use of such constraints in obtaining hardness from integrality gaps. 

\begin{theorem}\label{thm:lift}
 Given a flow sparsifier instance $\mathcal{H} = (G, k)$, there is a polynomial (in $n$ and $k$) time algorithm that outputs a flow sparsifier $H$ of quality $\alpha \leq \alpha'(\mathcal{H})$, where $\alpha'(\mathcal{H})$ is the quality of the best flow sparsifier that can be realized as a distributions over $0$-extensions.
\end{theorem}

\begin{proof}
We show that the following LP can give a flow-sparsifier with the desired properties:

\begin{eqnarray*}
  \min & \alpha \\
  \mbox{s.t} & \\
   & \begin{array}{l}
   \begin{array}{rcl}
   cong_G(\vec{w}) &\leq& \alpha \\
   w_{i,j} &=& \displaystyle\sum_{u,v \in V, u\neq v} c(u,v)x^{u,v}_{i,j} \qquad \forall i, j \in K\\
   \end{array} \\
   \left \{ \begin{array}{llll}
  x^{u,v}_{i,j} &=& x^{v,u}_{j,i} & \qquad \forall u, v\in V,u\neq v, i,j\in K\\
  \displaystyle\sum_{j\in K}x^{u,v}_{i,j} &=& x^u_i & \qquad \forall u, v\in V, u\neq v, i \in K\\
  \displaystyle\sum_{i\in K}x^u_i &=& 1 & \qquad \forall u \in V\\
  x^i_i &=& 1 &\qquad \forall i \in K\\
  x^{u,v}_{i,j} &\geq& 0 &\qquad \forall u, v \in V, u\neq v, i,j \in K\\
    \end{array} \right \} \mbox{\large \textbf{Earth-mover Constraints}}\\
    \end{array}
\end{eqnarray*}

\begin{lemma}
 The value of the LP is $\alpha \leq \alpha'(\mathcal{H})$.
\end{lemma}

\begin{proof}
Let $\mathcal{F}$ be the best distribution of 0-extensions. We explicitly give
a satisfying assignment for all the variables : 
\begin{eqnarray*}
\alpha &=& \alpha'(\mathcal{H})\\
x^u_i &=& \Pr_{f \sim \mathcal{F}}\left[f(u) = i\right]\\
x^{u,v}_{i,j} &=& \Pr_{f \sim \mathcal{F}}\left[f(u) = i \wedge
f(v) = j\right]\\
w_{i,j} &=& \sum_{u,v \in V, u\neq v} c(u,v)x^{u,v}_{i,j} 
\end{eqnarray*}
It's easy to see that the graph $H$ formed by $\set{w_{i,j}}$ is exactly the same as the flow sparsifier obtained from $\mathcal{F}$. So $H$ can be routed in $G$ with conjestion at most $\alpha'$. One can also verify that all the other constraints are satisfied. Thus, the value of the LP is at most $\alpha'(\mathcal{H})$.
\end{proof}

There are qualitatively two types of constraints that are associated with good flow-sparsifiers $H$: All flows routable in $H$ with congestion at most $1$ must be routable in $G$ with congestion at most $\alpha$. Actually, there is a notion of a hardest flow feasible in $H$ to route in $G$: the flow that saturates all edges in $H$ (i.e. $\vec{H}$). So the constraint that all flows routable in $H$ with congestion at most $1$ be also routable in $G$ with congestion at most $\alpha$ can be enforced by ensuring that $\vec{H}$ can be routed in $G$ with congestion at most $\alpha$. This constraint can be written using an infinite number of linear constraints on $H$ associated with the dual to a maximum concurrent flow problem, and in fact an oracle for the maximum concurrent flow problem can serve as a separation oracle for these constraints. 

The second set of constraints associated with good flow-sparsifiers are that all flows routable in $G$ with congestion at most $1$ can also be routed in $H$ with congestion at most $1$. This constraint can also be written as an infinite number of linear constraints on $H$, but no polynomial time separation oracle is known for these constraints. Instead, previous work relied on using oblivious routing guarantees to get an approximate separation oracle for this problem.

Intuitively, the constraint that all flows routable in $G$ can be routed in $H$ can be enforced in a number of ways. The strategy outlined in the preceding paragraph attempts to incorporate these constraints into the linear program. Alternatively, one could enforce that $H$ be realized as a distribution over $0$-extensions $G_f$. This would automatically enforce that all flows routable in $G$ would also be routable in $H$. However, this would require a variable for each $0$-extension $G_f$, and there would be exponentially many such variables. 

Yet the above linear programming formulation is a hybrid between these two approaches. In previous linear programming formulations, the sparsifier $H$ was not required to be explicitly generated from $G$, hence the need to enforce that it actually be a flow-sparsifer. When there is a variable for each $0$-extension, then $H$ is forced to be generated from $G$ and this constraint is implicitly satisfied. Yet just enforcing the Earth-mover constraints, as above, actually forces $H$ to have enough structure inherited from $G$ that $H$ is automatically a flow-sparsifier! This is the reason that we are able to get improved constructive results. To re-iterate, a simple lifting (corresponding to the Earth-mover constraints) does actually impose enough structure on $H$, that we can implicitly impose the constraint that $H$ be a flow-sparsifier without using exponentially many variables for each $0$-extension $G_f$!

\begin{lemma}
$\set{w_{i,j}:i,j \in K, i<j }$ is a flow sparsifier of quality $\alpha$. 
\end{lemma}

\begin{proof}
Let $H$ be the capacitated graph on $K$ formed by $\set{w_{i,j}}$. The LP system guarantees that $H$ can be routed in $G$ with conjestion at most $\alpha$, and thus we only need to show the other direction: every multi-commodity flow in $G$ with end points in $K$ can be routed in $H$ with conjestion at most 1.

Consider a multi-commoditiy flow $\set{f_{i,j}:i, j \in K, i < j}$ that can be routed in $G$. By the LP duality, we have $$\sum_{u < v}c(u,v)\delta(u,v) \geq \sum_{i < j}f_{i,j}\delta(i,j)$$ for every metric $\delta$ over $V$.

Let $\delta'$ be any metric over $K$, then

\begin{eqnarray*}
\sum_{i<j}\delta'(i,j)w_{i,j} = \sum_{i<j}\delta'(i,j)\sum_{u\neq
v}c(u,v)x^{u,v}_{i,j}=\sum_{u<v}c(u,v)\sum_{i\neq
j}x^{u,v}_{i,j}\delta'(i,j)\geq \sum_{u < v}c(u,v)EMD_{\delta'}(x^u, x^v).
\end{eqnarray*}

Define $\delta(u,v) = EMD_{\delta'}(x^u, x^v)$. Clearly, $\delta$ is a metric over $V$ and $\delta(i,j) = \delta'(i,j)$ for every $i, j \in K$. We have
 
\begin{eqnarray*}
\sum_{i<j}\delta'(i,j)w_{i,j} \geq \sum_{u < v}c(u,v)\delta(u,v)\geq
\sum_{i<j}f_{i,j}\delta(i,j) = \sum_{i<j}f_{i,j}\delta'(i,j)
\end{eqnarray*}

We have proved that $\sum_{i<j}\delta'(i,j)w_{i,j}\geq \sum_{i<j}f_{i,j}\delta'(i,j)$ for every metric $\delta'$ over $K$. By the LP duality, $f$ can be routed in $H$ with conjestion 1.
\end{proof}

\begin{lemma}
 The LP can be solved in polynomial (in $n$ and $k$) time.
\end{lemma}

\begin{proof}
 The LP contains polynomial number of variables and hence it is sufficient to give a separation oracle between a given point and the polytope defined by the LP. All constraints except whether or not $cong_G(\vec{H}) \leq \alpha$ can be directly checked, and for this remaining constraint the exact separation oracle is given by solving a maximum concurrent flow problem.
\end{proof}
\end{proof}


\section{Abstract Integrality Gaps and Rounding Algorithms}

In this section, we give a generalization of the hierarchical decompositions
constructed in \cite{R}. This immediately yields an $O(\log k)$-competitive
Steiner oblivious routing scheme, which is optimal. Also, from our hierarchical
decompositions we can recover the $O(\log k)$ bound on the flow-cut gap for
maximum concurrent flows given in \cite{LLR} and \cite{AR}. Additionally, we can
also give an $O(\log k)$ flow-cut gap for the maximum multiflow problem, which
was originally given in \cite{GVY}. This even yields an $O(\log k)$ flow-cut gap
for the relaxation for the requirement cut problem, which is given in \cite{NR}.
In fact, we will be able to give an abstract framework to which the results in
this section apply (and yield $O(\log k)$ flow-cut gaps for), and in this sense
we are able to help explain the intrinsic robustness of the worst-case ratio
between integral cover compared to fractional packing problems in graphs. 

Philosophically, this section aims to answer the question: Do we really need to
pay a price in the approximation guarantee for reducing to a graph on size $k$?
In fact, as we will see, there is often a way to combine both the reduction to a
graph on size $k$ and the rounding needed to actually obtain a flow-cut gap on
the reduced graph, into one step! This is exactly the observation that leads to
our improved approximation guarantee for Steiner oblivious routing. 

\subsection{$0$-Decomposition}

We extend the notion of $0$-extensions, which we previously defined, to a notion
of $0$-decompositions. Intuitively, we would like to combine the notion of a
$0$-extension with that of a decomposition tree. 

Again, given a $0$-extension $f$, we will denote $G_f$ as the graph on $K$ that
results from contracting all sets of nodes mapped to any single terminal. Then
we will use $c_f$ to denote the capacity function of this graph. 

\begin{definition}
Given a tree $T$ on $K$, and a $0$-extension $f$, we can generate a
$0$-decomposition $G_{f, T} = (K, E_{f, T})$ as follows:

The only edges present in $G_{f, T}$ will be those in $T$, and for any edge $(a,
b) \in E(T)$, let $T_a, T_b$ be the subtrees containing $a, b$ respectively that
result from deleting $(a, b)$ from $T$. 

Then $c_{f, T}(a,b)$ (i.e. the capacity assigned to $(a, b)$ in $G_{f, T}$ is:
$c_{f, T}(a,b) = \sum_{u, v \in K \mbox{ and } u \in T_a, v \in T_b} c_f(u, v)$.
\end{definition}

Let $\Lambda$ denote the set of $0$-extensions, and let $\Pi$ denote the set of trees on $K$.

\begin{claim}
For any distribution $\gamma$ on $\Lambda \times \Pi$, and for any demand $\vec{d} \in \Re^{K \choose 2}$, $cong_H(\vec{d}) \leq cong_G(\vec{d})$ where $H = \sum_{f \in \Lambda, T \in \Pi} \gamma(f, T) G_{f, T}$
\end{claim}

\begin{proof}
Clearly for all $f, T$, $\gamma(f, T) \vec{d}$ is feasible in $\gamma(f, T) G_f$ (because contracting edges only makes routing flow easier), and so because $G_{f, T}$ is a hierarchical decomposition tree for $G_f$, then it follows that $\gamma(f, T) \vec{d}$ is also feasible in $G_{f, T}$.
\end{proof}

\begin{claim}
Given any distribution $\gamma$ on $\Lambda \times \Pi$, let $H = \sum_{f \in \Lambda, T \in \Pi} \gamma(f, T) G_{f, T}$. Then $\sup_{\vec{d} \in \Re^{K \choose 2}} \frac{cong_G(\vec{d})}{cong_H(\vec{d})} = cong_G(\vec{H})$
\end{claim}

\begin{theorem} \label{thm:zdconst}
There is a polynomial time algorithm to construct a distribution $\gamma$ on $\Lambda \times \Pi$ such that $cong_G(\vec{H}) = O(\log k)$ where $H = \sum_{f \in \Lambda, T \in \Pi} \gamma(f, T) G_{f, T}$. 
\end{theorem}

We want to show that there is a distribution $\gamma$ on $\Lambda \times \Pi$ such that $cong_G(\vec{H}) = O(\log k)$. This will yield a generalization of \cite{R}. So as in \cite{M}, we set up a zero-sum game in which the first player chooses $f, T$ and plays $G_{f, T}$. The second player then chooses some metric space $d: K \times K \rightarrow \Re^+$ s.t. there is some extension of $d$ to a metric space on $V$ s.t. $\sum_{(u,v) \in E} d(u, v) c(u, v) \leq 1$. Then the first player loses $\sum_{(a, b)} c_{f, T}(a, b) d(a, b)$, which we will refer to as the cost of the metric space $d$ against $G_{f, T}$. 

It follows immediately from \cite{M} or \cite{LM} that a bound of $O(\log k)$ on the game value will imply our desired structural result. 

\begin{theorem}
The game value $\nu$ is $O(\log k)$
\end{theorem}

\begin{proof}

We consider an arbitrary strategy $\lambda$ for the second player, which is a
distribution on metric spaces $d$ that can be realized in $G$ with distance
$\times$ capacity units at most $1$. In fact, if we take the average metric
space $\Delta = \sum_{d} \lambda(d) d$, then this metric space can also be
realized in $G$ with at most $1$ unit of  distance $\times$ capacity. 

So we can bound the game value by showing that for all metric spaces $\Delta$ that can be realized with  distance $\times$ capacity units at most $1$, there is a $0$-decomposition $G_{f, T}$ for which the cost against $\Delta$ is at most $O(\log k)$. 

We can prove this by a randomized rounding procedure that is almost the same as the rounding procedure in \cite{FRT}: Scaling up the metric space, we can assume that all distances in the extension of $\Delta$ to a metric space on $V$ have distance at least $1$, and we assume $2^{\delta}$ is an upper bound on the diameter of the metric space. Then we need to first choose a $0$-extension $f$ for which the cost against $\Delta$ is $O(\log k)$ times the cost of realizing $\Delta$ in $G$. We do this as follows:

\vspace{0.5pc}

\begin{algorithmic}
\STATE Choose a random permutation $\pi(1), \pi(2), ..., \pi(k)$ of $K$

\STATE Choose $\beta$ uniformly at random from $[1, 2]$

\STATE $D_\delta \leftarrow \{V\}, i \leftarrow \delta -1$

\WHILE{ $D_{i+1}$ has a cluster which contains more than one terminal}

\STATE $\beta_i \leftarrow 2^{i - 1} \beta$

\FOR{ $\ell = 1$ to $k$}

\FOR{every cluster $S$ in $D_{i + 1}$}

\STATE Create a new cluster of all unassigned vertices in $S$ closer than $\beta_i$ to $\pi(\ell)$

\ENDFOR

\ENDFOR

\STATE $i \leftarrow i - 1$

\ENDWHILE

\end{algorithmic}

\vspace{0.5pc}

Then, exactly as in \cite{FRT}, we can construct a decomposition tree from the rounding procedure. The root node is $V$ corresponding to the partition $D_\delta$, and the children of this node are all sets in the partition $D_{\delta - 1}$. Each successive $D_i$ is a refinement of $D_{i + 1}$, so each set of $D_i$ is made to be a child of the corresponding set in $D_{i + 1}$ that contains it. At each level $i$ of this tree, the distance to the layer above is $2^{i}$, and one can verify that this tree-metric associated with the decomposition tree dominates the original metric space $\Delta$ restricted to the set $K$. Note that the tree metric does not dominate $\Delta$ on $V$, because there are some nodes which are mapped to the same leaf node in this tree, and correspondingly have distance $0$.  

If we consider any edge $(u, v)$, we can bound the expected distance in this tree metric from the leaf node containing $u$ to the leaf-node containing $v$. In fact, this expected distance is only a function of the metric space $\Delta$ restricted to $K \cup \{u, v\}$. Accordingly, for any $(u, v)$, we can regard the metric space that generates the tree-metric as a metric space on just $k + 2$ points. 

Formally, the rounding procedure in \cite{FRT} is:

\vspace{0.5pc}

\begin{algorithmic}
\STATE Choose a random permutation $\pi(1), \pi(2), ..., \pi(n)$ of $V$

\STATE Choose $\beta$ uniformly at random from $[1, 2]$

\STATE $D_\delta \leftarrow \{V\}, i \leftarrow \delta -1$

\WHILE{ $D_{i+1}$  has a cluster that is not a singleton}

\STATE $\beta_i \leftarrow 2^{i - 1} \beta$

\FOR{ $\ell = 1$ to $n$}

\FOR{every cluster $S$ in $D_{i + 1}$}

\STATE Create a new cluster of all unassigned vertices in $S$ closer than $\beta_i$ to $\pi(\ell)$

\ENDFOR

\ENDFOR

\STATE $i \leftarrow i - 1$

\ENDWHILE

\end{algorithmic}

\vspace{0.5pc}

Formally, \cite{FRT} proves a stronger statement than just that the expected distance (according to the tree-metric generated from the above rounding procedure) is $O(\log n)$ times the original distance. We will say that $u, v$ are split at level $i$ if these two nodes are contained in different sets of $D_i$. Let $X_i$ be the indicator variable for this event.

Then the distance in the tree-metric $\Delta_T$ generated from the above
rounding procedure is $\Delta_T(u, v) = \sum_i 2^{i + 1} X_i$. In fact,
\cite{FRT} proves the stronger statement that this is true even if $u, v$ are
not in the metric space (i.e. $u, v \notin V$) but are always grouped in the
cluster which they would be if they were in fact in the set $V$ (provided of
course that they can be grouped in such a cluster). More formally, we set $V' =
V \cup \{u, v\}$ and if the step "Create a new cluster of all unassigned
vertices in $S$ closer than $\beta_i$ to $\pi(\ell)$" is replaced with "Create a
new cluster of all unassigned vertices in V' in $S$ closer than $\beta_i$ to
$\pi(\ell)$", then \cite{FRT} actually proves in this more general context that 

$$\sum_{i} 2^{i + 1} X_i \leq O(\log n) \Delta(u, v)$$

When we input the metric space $\Delta$ restricted to $K$ into the above rounding procedure (but at each clustering stage we consider all of $V$) then we get exactly our rounding procedure. So then the main theorem in \cite{FRT} (or rather our restatement of it) is

(If $\Delta_T$ is the tree-metric generated from the above rounding procedure)

\begin{theorem} \cite{FRT}
For all $u, v$, $E[\Delta_T(u, v)] \leq O(\log k) \Delta(u, v)$.
\end{theorem}

So at the end of the rounding procedure, we have a tree in which each leaf correspond to a subset of $V$ that contains at most $1$ terminal. We are given a tree-metric $\Delta_T$ on $V$ associated with the output of the algorithm, and this tree-metric has the property that $\sum_{(u, v) \in E} c(u, v) \Delta_T(u, v) \leq O(\log k)$. 

We would like to construct a tree $T'$ from $T$ which has only leafs which
contain exactly one terminal. We first state a simple claim that will be
instructive in order to do this:

\begin{claim}~\label{claim:tree}
Given a tree metric $\Delta_T$ on a tree $T$ on $K$, $cost(G_{f, T}, \Delta_T) = cost(G_f, \Delta_T)$. 
\end{claim}

\begin{proof}
The graph $G_{f, T}$ can be obtained from $G_f$ by iteratively re-routing some edge $(a, b) \in E_f$ along the path connecting $a$ and $b$ in $T$ and adding $c_f(a, b)$ capacity to each edge on this path, and finally deleting the edge $(a, b)$. The original cost of this edge is $c(a, b) \Delta_T(a, b)$, and if $a = p_1, p_2, ..., p_r = b$ is the path connecting $a$ and $b$ in $T$, the cost after performing this operation is $c(a, b) \sum_{i = 1}^{r-1} \Delta_T(p_i, p_{i+1}) = c(a, b) \Delta_T(a, b)$ because $\Delta_T$ is a tree-metric. 
\end{proof}

We can think of each edge $(u, v)$ as being routed between the deepest nodes in the tree that contain $u$ and $v$ respectively, and the edge pays $c(u, v)$ times the distance according to the tree-metric on this path. Then we can perform the following procedure: each time we find a node in the tree which has only leaf nodes as children and none of these leaf nodes contains a terminal, we can delete these leaf nodes. This cannot increase the cost of the edges against the tree-metric because every edge (which we regard as routed in the tree) is routed on the same, or a shorter path. After this procedure is done, every leaf node that doesn't contain a terminal contains a parent $p$ that has a terminal node $a$. Suppose that the deepest node in the tree that contains $a$ is $c$ We can take this leaf node, and delete it, and place all nodes in the tree-node $c$. This procedure only affects the cost of edges with one endpoint in the leaf node that we deleted, and at most doubles th
 e cost paid by the edge because distances in the tree are geometrically decreasing. So if we iteratively perform the above steps, the total cost after performing these operations is at most $4$ times the original cost. 

And it is easy to see that this results in a natural $0$-extension in which each node $u$ is mapped to the terminal corresponding to the deepest node that $u$ is contained in. 

Each edge pays a cost proportional to a tree-metric distance between the
endpoints of the edge. So we know that $cost(G_f, \Delta_T) = O(\log k)$ because
the cost increased by at most a factor of $4$ from iteratively performing the
above steps. Yet using the above Claim, we get a $0$-extension $f$ and a tree
$T$ such that $cost(G_{f, T}, \Delta_T) = O(\log k)$ and because $\Delta_T$
dominates $\Delta$ when restricted to $K$, this implies that $cost(G_{f, T},
\Delta) \leq cost(G_{f, T}, \Delta_T) = O(\log k)$ and this implies the bound on
the game value. 

\end{proof}

In turn, using the arguments in \cite{LM}, implies:

\begin{theorem}~\label{thm:zdexist}
There is a distribution $\gamma$ on $\Lambda \times \Pi$ such that $cong_G(\vec{H}) = O(\log k)$ where $H = \sum_{f \in \Lambda, T \in \Pi} \gamma(f, T) G_{f, T}$. 
\end{theorem} 

Also, using the arguments in \cite{R} (because each $G_{f, T}$ is a tree and hence has a unique routing scheme), this gives us an $O(\log k)$-competitive Steiner oblivious routing scheme:

\begin{corollary}
Given $G = (V, E)$ and $K \subset V$, there is a set of unit flows for all $a, b
\in K$ that sends a unit flow from $a$ to $b$, such that given any demand
restricted to $K$, $\vec{d}$, the congestion incurred by this oblivious routing
scheme is $O(\log k)$ times the minimum congestion routing of $\vec{d}$. 
\end{corollary}

Actually, the above theorem can be made constructive directly using the techniques in \cite{R}, which build on \cite{Y}. We will not repeat the proof, instead we note the only minor difference in the proof. 

\begin{definition}
Let $\cR$ denote the set of pairs $(G_{f, T}, g)$ where $G_{f, T}$ is a $0$-decomposition of $G$, and $g$ is a function from edges in $G_{f, T}$ to paths in $g$ so that an edge $(a, b)$ in $G_{f, T}$ is mapped to a path connecting $a$ and $b$ in $G$. 
\end{definition}

Given a metric space $\delta$ on $V$, we can define the notion of the cost of a $(G_{f, T}, g)$ against $\delta$:

\begin{definition}
$$cost((G_{f, T}, g), \delta) = \sum_{(a, b) \in E(G_{f, T})} \sum_{(u, v) \in g(a,b)} c_{f,T}(a, b) \delta(u, v) = \sum_{(a, b) \in E(G_{f, T})}  c_{f, T}(a, b) \delta(g(a, b))$$
\end{definition}

\begin{corollary}~\label{cor:zdexist}
For any metric $\delta$ on $V$, there is some $(G_{f, T}, g) \in \cR$ such that:
$$cost((G_{f, T}, g), \delta) \leq O(\log k) \sum_{(u, v)} c(u,v) \delta(u, v)$$
\end{corollary}

\begin{proof}
We can apply Theorem~\ref{thm:zdexist} which implies that there is a distribution $\mu$ on $\cR$ s.t. for all edges $e \in E$, $$E_{(G_{f, T}, g) \leftarrow \mu}[\sum_{(a, b) \in E(G_{f, T}) \mbox{ s.t. } e \ni g(a, b)} c_{f,T}(a,b)] \leq O(\log k) c(e)$$ because we can take the optimal routing of $H = \sum_{f \in \Lambda, T \in \Pi} \gamma(f, T) G_{f, T}$ in $G$, which requires congestion at most $O(\log k)$ and if we compute a path decomposition of the routing schemes of each $G_{f, T}$ in the support of $\gamma$, we can use these to express the routing scheme as a convex combination of pairs from $\cR$. 
\end{proof}

So we can use an identical proof as in \cite{R} to actually construct a distribution $\gamma$ on $0$-decompositions s.t. for $H = \sum_{f \in \Lambda, T \in \Pi} \gamma(f, T) G_{f, T}$ we have  $cong_G(\vec{H}) = O(\log k)$. All we need to modify is the actual packing problem. In \cite{R}, the goal of the packing problem is to pack a convex combination of decomposition trees into the graph $G$ s.t. the expected relative load on any edge is at most $O(\log n)$. Here our goal is to pack a convex combination of $0$-decompositions into $G$. So instead of writing a packing problem over decomposition trees, we write a packing problem over pairs $(G_{f, T}, g) \in \cR$ and the goal is to find a convex combination of these pairs s.t. the relative load on any edge is $O(\log k)$. 

\cite{R} find a polynomial time algorithm by relating the change (when a decomposition tree is added to the convex combination) of the worst-case relative load (actually a convex function that dominates this maximum) to the cost of a decomposition tree against a metric. Analogously, as long as we can always (for any metric space $\delta$ on $V$) find a pair $(G_{f, T}, g)$ as in Corollary~\ref{cor:zdexist} an identical proof as in \cite{R} will give us a constructive version of Theorem~\ref{thm:zdexist}. And we can do this by again using the Theorem due to \cite{FHRT} (which we restated above in a more convenient notation for our purposes). This will give us a $0$-decomposition $G_{f, T}$ for which $\sum_{(a, b)} c_{f, T}(a,b) \delta(a, b) \leq O(\log k) \sum_{(u, v)} c(u,v) \delta(u, v)$ and we still need to choose a routing of $G_{f, T}$ in $G$. We can do this in a easy way: for each edge $(a, b)$ in $G_{f, T}$, just choose the shortest path according to $\delta$ connecting
  $a$ and $b$ in $G$. The length of this path will be $\delta(a, b)$, and so we have that $cost((G_{f, T}, g), \delta) \leq O(\log k) \sum_{(u, v)} c(u,v) \delta(u, v)$ as desired. Then using the proof in \cite{R} in our context, this immediately yields Theorem \ref{thm:zdconst}

\subsection{Applications}

Also, as we noted, this gives us an alternate proof of the main results in
\cite{LLR}, \cite{AR} and \cite{GVY}. We first give an abstract framework into
which these problems all fit:

\vspace{0.5pc}

\begin{theorem1}{Definition}{def:gpp}
We call a fractional packing problem $P$ a graph packing problem if the goal of the dual covering problem $D$ is to minimize the ratio of the total units of distance $\times$ capacity allocated in the graph divided by some monotone increasing function of the distances between terminals.
\end{theorem1}

\vspace{0.5pc}

Let $ID$ denote the integral dual graph covering problem. To make this definition seem more natural, we demonstrate that a number of
well-studied problems fit into this framework. 

\begin{example} \cite{LR}, \cite{LLR}, \cite{AR}
P: maximum concurrent flow; ID: generalized sparsest cut
\end{example}

Here we are given some demand vector $\vec{f} \in \Re^{K \choose 2}$, and the
goal is to maximize the value $r$ such that $r \vec{f}$ is feasible in $G$. Then
the dual to this problem corresponds to minimizing the total distance $\times$
capacity units, divided by $\sum_{(a, b)} \vec{f}_{a, b}
d(a, b)$, where $d$ is the induced semi-metric on $K$. The function in the denominator is
clearly a monotone increasing function of the distances between pairs of
terminals, and hence is an example of what we call a graph packing problem. The generalized sparsest cut problem corresponds to the "integral" constraint on the dual, that the distance function be a cut metric. 

\begin{example} \cite{GVY}
P: maximum multiflow; ID: multicut
\end{example}

Here we are given some pairs of terminals $T \subset { K \choose 2}$, and the goal is to find a flow $\vec{f}$ that can be routed in $G$ that maximizes $\sum_{(a, b) \in T } \vec{f}_{a, b}$. The dual to this problem corresponds to minimizing the total distance $\times$ capacity units divided by $\min_{(a, b) \in T} \{d(a, b) \} $, again where where $d$ is the induced semi-metric on $K$. Also the function in the denominator is again a monotone increasing function of the distances between pairs of terminals, and hence is another an example of what we call a graph packing problem. The multicut problem corresponds to the "integral" constraint on the dual that the distance function be a partition metric. 

\begin{example}
ID: Steiner multi-cut
\end{example}

\begin{example}
ID: Steiner minimum-bisection
\end{example}

\begin{example} \cite{NR}
P: multicast routing; ID: requirement cut
\end{example}

This is another partitioning problem, and the input is again a set of subsets $\{R_i\}_i$. Each subset $R_i$ is also given at requirement $r_i$, and the goal is to minimize the total capacity removed from $G$, in order to ensure that each subset $R_i$ is contained in at least $r_i$ different components. Similarly to the Steiner multi-cut problem, the standard relaxation for this problem  is to minimize the total amount of distance $\times$ capacity units allocated in $G$, s.t. for each $i$ the minimum spanning tree $T_i$ (on the induced metric on $K$)  on every subset $R_i$ has total distance at least $r_i$. Let $\Pi_i$ be the set of spanning trees on the subset $R_i$. Then we can again cast this relaxation in the above framework because the goal is to minimize the total distance $\times$ capacity units divided by $\min_i \{ \frac{ \min_{T \in \Pi_i} \sum_{(a, b) \in T} d(a, b)}{r_i} \} $. The dual to this fractional covering problem is actually a common encoding of multicast
  routing problems, and so these problems as well are examples of graph packing problems. Here the requirement cut problem corresponds to the "integral" constraint that the distance function be a partition metric.

In fact, one could imagine many other examples of interesting problems that fit into this framework. One can regard maximum multiflow as an \emph{unrooted} problem of packing an edge fractionally into a graph $G$, and the maximum concurrent flow problem is a rooted graph packing problem where we are given a fixed graph on the terminals (corresponding to the demand) and the goal is to pack as many copies as we can into $G$ (i.e. maximizing throughput). The dual to the Steiner multi-cut is more interesting, and is actually a combination of rooted and unrooted problems where we are given subset $R_i$ of terminals, and the goal is to maximize the total spanning trees over the sets $R_i$ that we pack into $G$. This is a combination of a unrooted (each spanning tree on any set $R_i$ counts the same) and a rooted problem (once we fix the $R_i$, we need a spanning tree on these terminals). 

Then any other flow-problem that is combinatorially restricted can also be seen to fit into this framework. 

As an application of our theorem in the previous section, we demonstrate that all graph packing problems can be reduced to graph packing problems on trees at the loss of an $O(\log k)$. So whenever we are given a bound on the ratio of the integral covering problem to the fractional packing problem on trees of say $C$, this immediately translates to an $O(C \log k)$ bound in general graphs. So in some sense, these embeddings into distributions on $0$-decompositions helps explain the intrinsic robustness of graph packing problems, and why the integrality gap always seems to be $O(\log k)$. In fact, since we can actually construct these distributions on $0$-decompositions, we obtain an {\sc Abstract Rounding Algorithm} that works for general graph packing problems. 

\vspace{0.5pc}

\begin{theorem1}{Theorem}{thm:agpp}
There is a polynomial time algorithm to construct a distribution $\mu$ on (a polynomial number of) trees on the terminal set $K$, s.t. $$E_{T \leftarrow \mu}[OPT(P, T)] \leq O(\log k) OPT(P, G)$$ and such that any valid integral dual of cost $C$ (for any tree $T$ in the support of $\mu$) can be immediately transformed into a valid integral dual in $G$ of cost at most $C$. 
\end{theorem1}

\vspace{0.5pc}

We first demonstrate that the operations we need to construct a $0$-decomposition only make the dual to a graph packing problem more difficult: Let $\nu(G, K)$ be the optimal value of a dual to a graph packing problem on $G = (V, E)$, $K \subset V$.

\begin{claim}
 Replacing any edge $(u, v)$ of capacity $c(u, v)$ with a path $u = p_1, p_2,
..., p_r = v$, deleting the edge $(u, v)$ and adding $c(u, v)$ units of capacity
along the path does not decrease the optimal value of the dual. 
\end{claim}

\begin{proof}
We can scale the distance function of the optimal dual so that the monotone increasing function of the distances between terminals is exactly $1$. Then the value of the dual is exactly the total capacity $\times$ distance units allocated. If we maintain the same metric space on the vertex set $V$, then the monotone increasing function of terminal distances is still exactly $1$ after replacing the edge $(u, v)$ by the path $u = p_1, p_2,
..., p_r = v$. However this replacement does change the cost (in terms of the total distance $\times$ capacity units). Deleting the edge reduces the cost by $c(u, v) d(u, v)$, and augmenting along the path increases the cost by $c(u, v) \sum_{i = 1}^{r-1} d(p_i, p_{i + 1})$ which, using the triangle inequality, is at least $c(u, v) d(u, v)$. 
\end{proof}

\begin{claim}
Suppose we join two nodes $u, v$ (s.t. not both of $u, v$ are terminals) into a new node $u'$, and replace each edge into $u$ or $v$ with a corresponding edge of the same capacity into $u'$. Then the optimal value of the dual does not decrease. 
\end{claim}

\begin{proof}
We can equivalently regard this operation as placing an edge of infinite capacity connecting $u$ and $v$, and this operation clearly does not change the set of distance functions for which the monotone increasing function of the terminal distances is at least $1$. And so this operation can only increase the cost of the optimal dual solution.
\end{proof}

We can obtain any $0$-decomposition $G_{f, T}$ from some combination of these operations. So we get that for any $f, T$:

\begin{corollary}
$\nu(G_{f, T}, K) \geq \nu(G, K)$
\end{corollary}

Let $\gamma$ be the distribution on $\Lambda \times \Pi$ s.t.  $H = \sum_{f \in \Lambda, T \in \Pi} \gamma(f, T) G_{f, T}$ and $cong_{G}(\vec{H}) \leq O(\log k)$. 

\begin{lemma}
$E_{(f, T) \leftarrow \gamma}[\nu(G_{f, T}, K)] \leq O(\log k) \nu(G, K)$.
\end{lemma}

\begin{proof}
We know that there is a metric $d$ on $V$ s.t. $\sum_{(u, v)} c(u, v) d(u, v) =
\nu(G, K)$ and that the monotone increasing function of $d$ (restricted to $K$)
is at least $1$. 

We also know that there is a simultaneous routing of each $\gamma(f, T) G_{f,
T}$ in $G$ so that the congestion on any edge in $G$ is $O(\log k)$. Then
consider the routing of one such $\gamma(f, T) G_{f, T}$ in this simultaneous
routing. Each edge $(a, b) \in E_{f, T}$ is routed to some distribution on paths
connecting $a$ and $b$ in $G$. In total $\gamma(f, T) c_f(a, b)$ flow is routed
on some distribution on paths, and consider a path $p$ that carries $C(p)$ total
flow from $a$ to $b$ in the routing of $\gamma(f, T) G_{f, T}$. If the total
distance along this path is $d(p)$, we increment the distance $d_{f, T}$ on the
edge $(a, b)$ in $G_{f, T}$ by $\frac{d(p) C(p)}{\gamma(f, T) c_{f, T}(a, b)}$,
and we do this for all such paths. We do this also for each $(a, b)$ in $G_{f,
T}$. 

If $d_{f, T}$ is the resulting semi-metric on $G_{f, T}$, then this distance function dominates $d$ restricted to $K$, because the distance that we allocate to the edge $(a, b)$ in $G_{f, T}$ is a convex combination of the distances along paths connecting $a$ and $b$ in $G$, each of which is at least $d(a, b)$. 

So if we perform the above distance allocation for each $G_{f, T}$, then each resulting $d_{f, T}, G_{f, T}$ pair satisfies the condition that the monotone increasing function of terminal distances ($d_{f, T}$) is at least $1$. But how much distance $\times$ capacity units have we allocated in expectation? 

$$E_{(f, T) \leftarrow \gamma}[\nu(G_{f, T}, K)] \leq \sum_{f, T} \gamma(f, T) \sum_{(a, b) \in E_{f, T}} c_{f, T}(a, b) d_{f, T}(a, b)$$

We can re-write

$$\sum_{f, T} \gamma(f, T) \sum_{(a, b) \in E_{f, T}} c_{f, T}(a, b) d_{f, T}(a,
b) = \sum_{(a, b) \in E} flow_{\vec{H}}(e) d(a, b) \leq cong_{G}(\vec{H})
\sum_{(a, b) \in E} c(a, b) d(a, b) \leq O(\log k)  \nu(G, K)$$
\end{proof}

And this implies:

\vspace{0.5pc}

\begin{theorem1}{Theorem}{thm:gpp}
For any graph packing problem $P$, the maximum ratio of the integral dual to the
fractional primal is at most $O(\log k)$ times the maximum ratio restricted to
trees. 
\end{theorem1}

\vspace{0.5pc}

And since we can actually construct such a distribution on $0$-decompositions in polynomial time, using Theorem~\ref{thm:zdconst}, this actually gives us an {\sc Abstract Rounding Algorithm}: We can just construct such a distribution on $0$-decompositions, sample one at random, apply a rounding algorithm to the tree to obtain a integral dual on the $0$-decomposition $G_{f, T}$ within $O(\log k)C$ times the value of the primal packing problem on $G$. This integral dual on the $0$-decomposition $G_{f, T}$ can then be easily mapped back to an integral dual on $G$ at no additional cost precisely because we can set the distance in $G$ of any edge $(a, b)$ to be the tree-distance according to the integral dual on $G_{f, T}$ between $a$ and $b$. Using Claim~\ref{claim:tree}, this implies that the cost of the dual in $G_f$ is equal the cost of the dual in $G_{f, T}$. And we can choose an integral dual $\delta'$ in $G$ in which for all $u, v$, $\delta'(u, v) = \delta(f(u), f(v))$ and 
 the cost of this dual $\delta'$ on $G$ is exactly the cost of $G_f$ on $\delta$. And so we have an integral dual solution in $G$ of cost at most $O(\log k) C$ times the cost of the fractional primal packing value in $G$, where $C$ is the maximum integrality gap of the graph packing problem restricted to trees. This yields our {\sc Abstract Rounding Algorithm}:

\vspace{0.5pc}

\begin{theorem1}{Theorem}{thm:agpp}
There is a polynomial time algorithm to construct a distribution $\mu$ on (a polynomial number of) trees on the terminal set $K$, s.t. $$E_{T \leftarrow \mu}[OPT(P, T)] \leq O(\log k) OPT(P, G)$$ and such that any valid integral dual of cost $C$ (for any tree $T$ in the support of $\mu$) can be immediately transformed into a valid integral dual in $G$ of cost at most $C$. 
\end{theorem1}

\vspace{0.5pc}

\begin{corollary}
If there is a $C$-approximation algorithm for a graph partitioning problem restricted to trees, then there is an $O(C\log k)$ approximation algorithm for the graph partitioning problem in general graphs.
\end{corollary}

So, there is a natural, generic algorithm  associated with this theorem : 
\begin{algorithmic}[1]
\STATE Decompose $G$ into an $O(\log k)$-oblivious distribution $\mu$ of 0-decompostion trees;
\STATE Randomly select a tree $G_{f,\tau}$ from the distribution $\mu$;
\STATE Solve the problem on the tree $G_{f, \tau}$, let $\delta$ be the metric the algorithm output;
\STATE Return $(\delta, f)$.
\end{algorithmic}

 For example, this gives a generic algorithm that achieves an $O(\log k)$ guarantee for \emph{both} generalized sparsest cut and multicut. The previous techniques for rounding a fractional solution to generalized sparsest cut \cite{LLR}, \cite{AR} rely on metric embedding results, and the techniques for rounding fractional solutions to multicut \cite{GVY} rely on purely combinatorial, region-growing arguments. Yet, through this theorem, we can give a unified rounding algorithm that achieves an $O(\log k)$ guarantee for both 
 of these problems, and more generally for graph packing problems (whenever the integrality gap restricted to trees is a constant). 




\newpage


\section{Acknowledgments}

We would like to thank Swastik Kopparty, Ryan O'Donnell and Yuval Rabani for many helpful discussions.

\nocite{JLS}

\bibliographystyle{latex4}
\bibliography{latex4}


\newpage

\appendix
\section{Harmonic Analysis}\label{sec:aha}

We consider the group $F_2^d = \{-1, +1\}^d$ equipped with the group operation $s \circ t = [s_1 * t_1, s_2 * t_2, ... s_d * t_d] \in F_2^d$. Any subset $S \subset [d]$ defines a character $\chi_S(x) = \prod_{i \in S} x_i : F_2^d \rightarrow \{-1, +1\}$. See \cite{O} for an introduction to the harmonic analysis of Boolean functions. 

Then any function $f: \{-1, +1\}^d \rightarrow \Re$ can be written as:
$$f(x) = \sum_S \hat{f}_S \chi_S(x)$$

\begin{fact}
For any $S, T \subset [d]$ s.t. $S \neq T$, $E_x [\chi_S(x) \chi_T(x)] = 0$
\end{fact}

For any $p > 0$, we will denote the $p$-norm of $f$ as $||f||_p = \Big( E_x[f(x)^p] \Big)^{1/p}$. Then

\begin{theorem}
[Parseval]
$$\sum_S \hat{f}_S^2 = E_x[f(x)]^2] = ||f||_2^2$$
\end{theorem}

\begin{definition}
Given $-1 \leq \rho \leq 1$, Let $y \sim_\rho x$ denote choosing $y$ depending on $x$ s.t. for each coordinate $i$, $E[y_i x_i] = \rho$. 
\end{definition}

\begin{definition}
Given $-1 \leq \rho \leq 1$, the operator $T_\rho$ maps functions on the Boolean cube to functions on the Boolean cube, and for $f:  \{-1, +1\}^d \rightarrow \Re$, $T_\rho(f(x)) = E_{y \sim_\rho x} [f(y)]$. 
\end{definition}

\begin{fact}
$T_\rho(\chi_S(x)) = \chi_S(x) \rho^{|S|}$
\end{fact}

In fact, because $T_\rho$ is a linear operator on functions, we can use the Fourier representation of a function $f$ to easily write the effect of applying the operator $T_\rho$ to the function $f$:

\begin{corollary}
$T_\rho(f(x)) = \sum_S \rho^{|S|} \hat{f}_S \chi_S(x)$
\end{corollary}

\begin{definition}
The \emph{Noise Stability} of a function $f$ is $NS_{\rho}(f) = E_{x, y \sim_{\rho} x}[f(x) f(y)]$
\end{definition}

\begin{fact}
$NS_\rho(f) = \sum_S \rho^{|S|} \hat{f}_S^2$
\end{fact}

\begin{theorem}
[Hypercontractivity] \cite{Bon} \cite{Bec}
For any $q \geq p \geq 1$, for any $\rho \leq \sqrt{\frac{p-1}{q-1}}$
$$||T_\rho f||_q \leq ||f||_p$$
\end{theorem}

A statement of this theorem is given in \cite{O} and \cite{DFKO} for example.

\begin{definition}
A function $g: \{-1, +1\}^d \rightarrow \Re$ is a $j$-junta if there is a set $S \subset [d]$ s.t. $|S| \leq j$ and $g$ depends only on variables in $S$ - i.e. for any $x, y \in F_2^d$ s.t. $\forall_{i \in S} x_i = y_i$ we have $g(x) = g(y)$. We will call a function $f$ an $(\epsilon, j)$-junta if there is a function $g: \{-1, +1\}^d \rightarrow \Re$ that is a $j$-junta and $Pr_x[f(x) \neq g(x)] \leq \epsilon$.
\end{definition}

We will use a quantitative version of Bourgain's Junta Theorem \cite{Bou} that is given by Khot and Naor in \cite{KN}:

\begin{theorem}
[Bourgain] \cite{Bou}, \cite{KN}
Let $f \{-1, +1\}^d \rightarrow \{-1, +1\}$ be a Boolean function. Then fix any $\epsilon, \delta \in (0, 1/10)$.  Suppose that
$$\sum_S (1-\epsilon)^{|S|} \hat{f}_S^2 \geq 1 - \delta$$
then for every $\beta > 0$, $f$ is a
$$\Big (2^{c \sqrt{\log 1/\delta \log \log 1/\epsilon}}\Big ( \frac{\delta}{\sqrt{\epsilon}} + 4^{1/\epsilon} \sqrt{\beta}\Big), \frac{1}{\epsilon \beta} \Big ) \mbox{-junta}$$
\end{theorem}

This theorem is often described as mysterious, or deep, and has lead to some breakthrough results in theoretical computer science \cite{KN}, \cite{KV} and is also quite subtle. For example, this theorem crucially relies on the property that $f$ is a Boolean function, and in more general cases only much weaker bounds are known \cite{DFKO}.

\end{document}